\documentclass[journal,draftcls,onecolumn,12pt,twoside]{IEEEtran}

\usepackage{amssymb, amsthm, enumerate, array, graphicx, changepage}
\usepackage[utf8]{inputenc} 
\usepackage[T1]{fontenc}
\usepackage{url}
\usepackage{ifthen}
\usepackage{cite}
\usepackage[cmex10]{amsmath} 
\usepackage{tcolorbox}
\usepackage{multirow}

\newtheorem{theorem}{Theorem}
\newtheorem{proposition}{Proposition}
\newtheorem{corollary}{Corollary}
\newtheorem{lemma}{Lemma}
\newtheorem{definition}{Definition}
\newtheorem{example}{Example}

\newtheorem{remark}{Remark}

\newcommand{\CAC}{\mathrm{CAC}}
\newcommand{\CACe}{\mathrm{CAC^e}}
\newcommand{\HH}{\mathsf{H}}
\newcommand{\C}{\mathcal{C}}
\newcommand{\Z}{\mathbb{Z}}

\begin{document}
\title{Optimal Constant-Weight and Mixed-Weight Conflict-Avoiding Codes} 

\author{
{Yuan-Hsun Lo,~\IEEEmembership{IEEE~Member}, Tsai-Lien Wong, Kangkang Xu,\\ Yijin Zhang,~\IEEEmembership{IEEE~Senior Member}}
\thanks{This article was presented in part at the 2024 IEEE International Symposium on Information Theory.}
\thanks{
This work was supported in part by the National Science and Technology Council, Taiwan, under Grants 112-2115-M-153-004-MY2 and 113-2115-M-110-003-MY2, and in part by the National Natural Science Foundation of China under Grant 62071236. \textit{(Corresponding author: Yijin Zhang)}
}
\thanks{Y.-H. Lo is with Department of Applied Mathematics, National Pingtung University, Taiwan.  Email: yhlo0830@gmail.com}
\thanks{T.-L. Wong and K. Xu are with Department of Applied Mathematics, National Sun Yat-sen University, Taiwan. Email: tlwong@math.nsysu.edu.tw, ivykkxu107@gmail.com}
\thanks{Y. Zhang is with Nanjing University of Science and Technology, Nanjing 210094, China. Email: yijin.zhang@gmail.com}
}

\maketitle

\begin{abstract}
A conflict-avoiding code (CAC) is a deterministic transmission scheme for asynchronous multiple access without feedback.
When the number of simultaneously active users is less than or equal to $w$, a CAC of length $L$ with weight $w$ can provide a hard guarantee that each active user has at least one successful transmission within every consecutive $L$ slots.
In this paper, we generalize some previously known constructions of constant-weight CACs, and then derive several classes of optimal CACs by the help of Kneser's Theorem and some techniques in Additive Combinatorics.
Another spotlight of this paper is to relax the identical-weight constraint in prior studies to study mixed-weight CACs for the first time, for the purpose of increasing the throughput and reducing the access delay of some potential users with higher priority.
As applications of those obtained optimal CACs, we derive some classes of optimal mixed-weight CACs.
\end{abstract}


\section{Introduction}\label{sec:intro}

A conflict-avoiding code (CAC)~\cite{LT05} is a deterministic grant-free scheme for asynchronous multiple access without feedback.
Unlike probabilistic schemes, a CAC can offer a hard guarantee on the worst-case delay relying on its good cross-correlation property.
This hard guarantee is desirable to provide satisfactory services for many mission-critical applications with ultra-reliable and low-latency communications (URLLC)~\cite{Bennis18PIEEE}, such as industrial automation, intelligent transportation, telemedicine, and Meta-Universe.
Other deterministic schemes, like protocol sequences~\cite{LWLZCX23}, rendezvous sequences~\cite{CZWZYL22}, can also been seen as variants of CACs for other performance guarantees. 

Following~\cite{LT05}, this paper considers the collision channel model without feedback~\cite{MM_85}.
The time axis is partitioned into equal-length time slots, whose duration corresponds to the transmission time for one packet.
Assume that there is no global time synchronization among the users and no feedback information from the receiver.
So, each user $i$ has a relative time offset $\tau_i$ in unit of a slot, which is random but remains fixed throughout the communication session.
In a slot, if two or more than two users are transmitting packets simultaneously, then a collision occurs and all of the packets are lost; otherwise, the packet transmitted by a unique user can be received successfully.
Let $\mathcal{C}$ be a CAC of length $L$ with weight $w$.
Each codeword in $\mathcal{C}$ consists of $w$ elements $x_1,x_2,\ldots,x_w$, where $0\leq x_i\leq L-1$.
Each user is preassigned a unique codeword from $\mathcal{C}$, and a user $i$ sends out a packet at slot $t+\tau_i$ if and only if the user $i$ is active and the corresponding codeword contains an integer $x_i=t+\tau_i$ (mod $L$).

By applying a CAC of length $L$ with weight $w$ to access, any two active users have at most one collision between them in a period of $L$ slots no matter what the time offsets are.
This property guarantees that each active user has at least one successful transmission in a period of $L$ slots if there are at most $w$ active users at the same time.
The design goal of CACs is to maximize the number of codewords (i.e., the number of potential users that can support) for given $L$ and $w$.
Note that the CACs design for the slot-synchronous model can be extended to that for the fully asynchronous model~\cite{MM_85}.

\subsection{Conflict-Avoiding Codes}

Let $\mathbb{Z}_L\triangleq\{0,1,\ldots,L-1\}$ denote the ring of residue modulo $L$, and let $\mathbb{Z}_L^*\triangleq\mathbb{Z}_L\setminus\{0\}$.
For $S\subseteq\mathbb{Z}_L$, let 
\begin{equation}\label{eq:difference-set-def}
d^*(S)\triangleq\{a-b\, (\text{mod } L):\,a,b\in S, a\neq b\}
\end{equation}
denote the set of \emph{(nonzero) differences} of $S$.

\begin{definition}\label{def:CAC}\rm
Let $L$ and $w$ be two positive integers with $L>w$. 
A \textit{conflict-avoiding code (CAC)} $\mathcal{C}$ of length $L$ with weight $w$ is a collection of $w$-subsets, called codeword, of $\mathbb{Z}_L$ such that
\begin{equation}\label{eq:CAC}
d^*(S)\cap d^*(S')=\emptyset \quad \forall S,S'\in\mathcal{C}, S\neq S'.
\end{equation}
\end{definition}

The condition in \eqref{eq:CAC} is called the \textit{disjoint-difference-set} property.
Without loss of generality, we may assume that all codewords contain $0$.
Let $\CAC(L,w)$ denote the class of all CACs of length $L$ with weight $w$.
The maximum size of a code in $\CAC(L,w)$ is denoted by $K(L,w)$, i.e.,
\begin{align*}
K(L,w)\triangleq\max\{|\C|:\,\C\in\CAC(L,w)\}.
\end{align*}
A code $\C\in\CAC(L,w)$ is called \textit{optimal} if its code size achieves $K(L,w)$.
As $\bigcup_{S\in\C}d^*(S)\subseteq\mathbb{Z}_L^*$ for $\C\in\CAC(L,w)$, an optimal code $\C\in\CAC(L,w)$ is said to be \textit{tight} if $\bigcup_{S\in\C}d^*(S)=\mathbb{Z}_L^*$.

A $w$-subset $S\subseteq\mathbb{Z}_L$ is said to be \textit{equi-difference} with \textit{generator} $g\in\Z_L^*$ if $S$ is of the form $\{0,g,2g,\ldots,(w-1)g\}$.
Observe that $d^*(S)=\{\pm g,\pm 2g,\ldots,\pm (w-1)g\}$ and $|d^*(S)|\leq 2w-2$ if $S$ is an equi-difference codeword with generator $g$.
Note that $S$ is called \textit{exceptional} if $|d^*(S)|$ is strictly less than $2w-2$.
A CAC is called equi-difference if it entirely consists of equi-difference codewords.
Let $\CACe(L,w)\subset\CAC(L,w)$ denote the class of all equi-difference codes and $K^e(L,w)$ be the maximum size among $\CACe(L,w)$.
Obviously, $K^e(L,w)\leq K(L,w)$.

For fixed $w$, it was shown in~\cite{SWC10} that $K(L,w)$ increases approximately with slope $(2w-2)^{-1}$ as a function of length $L$, and meanwhile, an asymptotically upper bound of $K(L,w)$ was given in~\cite{SW10}.
Based on some finite-field properties, some constructions of CACs for general weights can be found in~\cite{MMSJ07}, together with a series of optimal CACs with weight $w=4,5$.
For small $w$, the exact value of $K(L,3)$ is completely determined by \cite{LT05,JMJTT07,MFU09,FLM10} for even $L$.
As for odd length, $K(L,3)$ is determined for $L$ being some particular prime~\cite{LT05} and some composite number with particular factors~\cite{Levenshtein07,FLS14,MZS14,MM17}.
If only equi-difference codewords are concerned, $K^e(L,w)$ is obtained for some particular $L$ with $w=3$ in \cite{WF13,LMSJ14} and with weight $w=4$ in \cite{LFL15,LMJ16}.
In the case of tight CACs, \cite{Momihara07} presented a necessary and sufficient condition for the existence of tight equi-difference CACs of weight $3$, which was rewritten in the notion of multiplicative order of $2$ in \cite{FLS14}.


\subsection{Known Optimal Constant-Weight CACs}



We shall recall some previously known results on optimal CACs provided in literature.
The first two ones are based on the theory of quadratic residues.

Given a positive integer $n$, a nonzero element $a\in\mathbb{Z}_n$ is called a \textit{quadratic residue} if there exists an integer $x\in\mathbb{Z}_n$ such that $a=x^2$; otherwise, $a$ is called a \textit{quadratic non-residue}.
Consider an odd prime $p$.
The \textit{Legendre symbol} on $\mathbb{Z}_p$ is defined (e.g.,~\cite{IR90}) as, for $a\in\mathbb{Z}_p$, 
\begin{align*}
\left(\frac{a}{p}\right)\triangleq
\begin{cases}
1 & \text{if } a \text{ is a quadratic residue modulo }p, \\
-1 & \text{if } a \text{ is a quadratic non-residue modulo }p, \\
0 & \text{if } a=0.
\end{cases}
\end{align*}
It can be shown that the Legendre symbol is multiplicative:
\begin{equation}\label{eq:Legendre}
\left(\frac{ab}{p}\right) = \left(\frac{a}{p}\right)\left(\frac{b}{p}\right).
\end{equation}

The following result is given in~\cite[Theorems 3 and 7]{SW10}, which plays an important role in the derivation of a tight asymptotic upper bound on $K(L,w)$.

\begin{theorem}[\cite{SW10}]\label{thm:known_w-1p}
Let $p$ be an odd prime and $w$ be an integer such that $2\leq w\leq p$.
If 
\begin{equation}\label{eq:QR1}
\left(\frac{-1}{p}\right)=-1
\end{equation}
and
\begin{equation}\label{eq:QR2}
\left(\frac{i}{p}\right)\left(\frac{i-w+1}{p}\right)=-1,\ \forall i=1,2,\ldots,w-2, 
\end{equation}
then there exists a code in $\CACe((w-1)p,w)$ with $(p-1)/2$ codewords.
In particular, if $w-1$ is an odd prime such that $p\geq 2w-1$, then
\begin{align*}
K\left((w-1)p,w\right)=\frac{p-1}{2}.
\end{align*}
\end{theorem}

The following is an adaptation of a construction given in~\cite[Theorem 5.1]{MMSJ07}.



\begin{theorem}[\cite{MMSJ07}]\label{thm:known_wp}
Let $p$ be a prime such that $\gcd(p,w)=1$. 
If there is a code in $\CACe(p,w)$ with $m$ codewords and
\begin{align*}
\left(\frac{i}{p}\right)\left(\frac{i-w}{p}\right)=-1,\ \forall i=1,2,\ldots,w-1, 
\end{align*}
then there exists a code in $\CACe(wp,w)$ with $m+\frac{p-1}{2}+1$ codewords.
\end{theorem}

By Theorem~\ref{thm:known_wp}, \cite{MMSJ07} obtained optimal $\C\in\CAC(4p,4)$ with $K(4p,4)=|C|=\frac{p-1}{6}+\frac{p-1}{2}+1$, where $p=13\ (\bmod\, 24)$ and satisfies some particular conditions.

The last one is about a construction given in~\cite[Theorem 13]{SWC10}.

\begin{theorem}[\cite{SWC10}]\label{thm:known_2w-1p}
Let $p$ be a prime number such that $p>2w-1$.
If there is a code in $\mathrm{CAC^e}(p,w)$ with $m$ codewords, then there exists a code in $\CACe((2w-1)p,w)$ with $p+m$ codewords.
In particular, if $p-1$ is divisible by $2w-2$ and $m=(p-1)/(2w-2)$, then
\begin{align*}
K\left((2w-1)p,w\right)=p+\frac{p-1}{2w-2}.
\end{align*}
\end{theorem}

\subsection{Mixed-Weight Conflict-Avoiding Codes}

A CAC assumes that all the users have the same number of transmission opportunities under the same throughput/delay performance requirement.
However, in heterogeneous systems~\cite{TF15,CLW19} with different individual performance requirements, users may be divided into several groups according to their \emph{priority}: users with higher priority should have higher probability to successfully transmit their packets in order to increase their throughput and reduce their access delay.
Motivated by this heterogeneity, we propose a generalization of CACs, called \emph{mixed-weight CACs} by increasing the weights of some codewords, so that those users assigned with larger-weight codewords of length $L$ are able to transmit more packets successfully during every $L$ consecutive slots and enjoy smaller worst-case delay within $L$ slots, which will be analyzed after Theorem~\ref{thm:mixed-2w-1pr-optimal}.

\begin{definition}\label{def:mixCAC}\rm
Let $L$ be a positive integer and $\mathcal{W}$, called \textit{weight-set}, be a set of positive integers.
A \emph{mixed-weight CAC} $\mathcal{C}$ of length $L$ with weight-set $\mathcal{W}$ is a collection of subsets of $\mathbb{Z}_L$ such that, (i) each subset is of size in $\mathcal{W}$; and (ii) $\mathcal{C}$ satisfies the disjoint-difference-set property as shown in~\eqref{eq:CAC}.
\end{definition}

Let $\text{CAC}(L,\mathcal{W})$ denote the class of all mixed-weight CACs of length $L$ with weight-set $\mathcal{W}$.
Similar to the design goal of CACs, the problem of mixed-weight CACs aims to maximize the total number of codewords that can be supported, when $L$ and $\mathcal{W}$ are given.
However, as the number of high priority users is relatively smaller than the others, it would be meaningful to maximize the number of low priority users when the numbers of high priority ones are fixed.

Let $\mathcal{W}^*=\{w^*_1,\ldots,w^*_t\}$ and $\mathcal{W}$ be two sets of positive integers with $w^*_1>\cdots> w^*_t> w,\, \forall w\in\mathcal{W}$. 
For a $t$-tuple with non-negative integers $\boldsymbol{n}=(n_1,\ldots,n_t)$, denote by $K(L,\mathcal{W};\mathcal{W}^*,\boldsymbol{n})$ be the maximum size of some code in $\text{CAC}(L,\mathcal{W}^*\cup\mathcal{W})$, in which the number of codewords with weight $w^*_i$ is exactly $n_i$ for all $i$.
A code $\mathcal{C}\in\text{CAC}(L,\mathcal{W}^*\cup\mathcal{W})$ is called \textit{optimal} if $|\mathcal{C}|=K(L,\mathcal{W};\mathcal{W}^*,\boldsymbol{n})$ and agrees the size-constraint of $w^*_i$-weight codewords, for each $i$.
We simply denote by $K(L,w;w^*,n)$ when $\mathcal{W}=\{w\}$ and $\mathcal{W}^*=\{w^*\}$.

\subsection{Main Contributions}

The considered length in this paper is of the form $L=ap^r$, where $\gcd(a,p)=1$ for some $a$ and some prime $p$.
Since $\gcd(a,p)=1$, we have $\mathbb{Z}_{ap^r}\cong\mathbb{Z}_{a}\times\mathbb{Z}_{p^r}$.
A natural bijection between $\mathbb{Z}_{ap^r}$ and $\mathbb{Z}_{a}\times\mathbb{Z}_{p^r}$ is via the Chinese Remainder Theorem (CRT)~\cite{IR90}, i.e., $\theta:\mathbb{Z}_{ap^r}\to\mathbb{Z}_{a}\times\mathbb{Z}_{p^r}$ by 
\begin{equation}\label{eq:CRT-correspondence}
\theta(x)=(x\,(\bmod\, a), x\,(\bmod\, p^r)).
\end{equation}
Therefore, a $w$-subset in $\mathbb{Z}_{ap^r}$ can be simply put as a $w$-subset in $\mathbb{Z}_{a}\times\mathbb{Z}_{p^r}$.

In this paper, we will provide various direct constructions to obtain CACs, say Theorems~\ref{thm:construction_direct}, \ref{thm:construction_wpr} and \ref{thm:construction_2w-1pr}, by extending the applicable code length from $ap$ to $ap^r$ for some odd prime $p$ and any positive integer $r$.
As elements in $\Z_{ap^r}$ can be viewed as ordered pairs in $\Z_a\times\Z_{p^r}$ via the CRT correspondence, the main idea is to express an element in $\Z_{p^r}$ as its $p$-ary representation.
As such, a $w$-subset in $\Z_{ap}$ can be raised to various $w$-subsets in $\Z_{ap^r}$ according to which layer the $w$-subsets lying one.
The concept ``layer'' will be defined in Section~\ref{sec:CAC_direct_p}.
The disjoint-difference-set property is then confirmed by some nice algebraic property of the $p$-ary representation, say Proposition~\ref{prop:p-ary}, involving the invertible elements in $\Z^*_{p^r}$.

By the help of some results in Additive Combinatorics, the sufficient conditions when the obtained CACs in Theorems~\ref{thm:construction_direct}, \ref{thm:construction_wpr} and \ref{thm:construction_2w-1pr} being optimal are characterized in Theorems~\ref{thm:main_w-1dpr}, \ref{thm:main_wpr} and \ref{thm:main_2w-1pr}, which are the generalizations of Theorems~\ref{thm:known_w-1p}, \ref{thm:known_wp} and \ref{thm:known_2w-1p}, respectively.
More precisely, as the evaluation of the number of exceptional codewords plays a key role in deriving upper bounds on the maximum size of a CAC, we study the stabilizers of the difference set of an exceptional codeword in more detail by applying Kneser's theorem carefully.
Let $H$ be the set of stabilizers of $d^*(S)\cup\{0\}$ for some exceptional $w$-subset $S$ in $\mathbb{Z}_L$.
We first derive a bound on $|H|$ in Corollary~\ref{cor:excpetional-bound}, and then claim that $|H|$ does not divide $w-1$ nor $2w-1$ in Lemma~\ref{lem:nonexceptional}.
This key lemma leads to a tight upper bound on $K(ap^r,w)$ for some $a$ and some prime $p$, since it can rule out many potential exceptional codewords. 
It is worth noting that the use of Kneser's theorem in deriving upper bounds on the maximum size of a CAC can be found in literature.
See the proofs of Theorem~\ref{thm:known_w-1p} and Theorem~\ref{thm:known_2w-1p}, i.e., \cite[Theorems 3 and 7]{SW10} and \cite[Theorem 13]{SWC10}, for instance.
It should be noted that this is the motivation of our revisit of Kneser's theorem in this paper.
However, instead of investigating the set of stabilizers of the difference set of an exceptional codeword, the prior works focused on studying the difference set itself, which loses some essential information and results in some additional conditions should be added to guarantee the tightness of the upper bounds they obtained.
For example, Lemma~\ref{lem:nonexceptional} can relax the two conditions that $w-1$ is an odd prime and $p\geq 2w-1$ for the optimality in Theorem~\ref{thm:known_w-1p}.
See Table~\ref{tab:comparison} for a comparison of our results and the corresponding previously known results.

\begin{table}[h]
\begin{center}
\begin{tabular}{l|l|l}
Reference & Applicable Length & Conditions for Optimality \\ \hline \hline
Theorem~\ref{thm:known_w-1p} \cite{SW10} & $L=(w-1)p$ & $p\geq 2w-1$ and $w-1$ is an odd prime \\
Theorem~\ref{thm:main_w-1dpr} & $L=\frac{w-1}{d}p^r$, $\forall r\geq 1$ and $d|(w-1)$ & $p\geq w$, $d|(w-1)$, and $2d|(p-1)$ \\ \hline
Theorem~\ref{thm:known_wp} \cite{MMSJ07} & $L=wp$ & N/A \\ 
Theorem~\ref{thm:main_wpr} & $L=wp^r$, $\forall r\geq 1$ & $m=(p-1)/(2w-2)\geq 1$ \\ \hline
Theorem~\ref{thm:known_2w-1p} \cite{SWC10} & $L=(2w-1)p$ & $m=(p-1)/(2w-2)>1$ \\ 
Theorem~\ref{thm:main_2w-1pr} & $L=(2w-1)p^r$, $\forall r\geq 1$ & $m=(p-1)/(2w-2)\geq 1$ \\ \hline 
\end{tabular}
\end{center}
\caption{A comparison between previously known results and ours.}
\label{tab:comparison}
\end{table}

Finally, as applications of those obtained optimal CACs, we derive some classes of optimal mixed-weight CACs.

Here is the summary of our contribution.
\begin{enumerate}[1.]
\item Generalize Theorem~\ref{thm:known_w-1p} in two aspects: (i) Extend the applicable length $L=(w-1)p$ to $\frac{w-1}{d}p^r$, for any factor $d$ of $w-1$ and integer $r\geq 1$; and (ii) Remove the conditions that $p\geq 2w-1$ and $w-1$ is an odd prime when the equality holds.
\item As an application of Theorem~\ref{thm:main_w-1dpr}, we obtain constructions of optimal CACs: 
	\begin{enumerate}[(i)]
	\item $\C\in\CAC((w-1)p^r,w)$ with $|\C|=(p^r-1)/2$ for infinitely many primes $p$. See Corollary~\ref{coro:w-1prw} for some examples with $4\leq w\leq 11$.
	\item $\C\in\CAC(2p^r,5)$ with $|\C|=(p^r-1)/4$ for all primes $p\equiv 5\ (\bmod\ 24)$.
	\item $\C\in\CAC(3p^r,7)$ with $|\C|=(p^r-1)/2$ for all primes $p\equiv 5\ (\bmod\ 8)$ with $10^{(p-1)/4}\equiv 1\ (\bmod\ p)$.
	\end{enumerate}
\item Extend the applicable length $L$ of Theorem~\ref{thm:known_wp} (resp., Theorem \ref{thm:main_wpr}) from $wp$ (resp., $(2w-1)p$) to $wp^r$ (resp., $(2w-1)p^r$), for any integer $r\geq 1$. 
In particular, we provide a sufficient condition for the constructed CACs of length $wp^r$ to be optimal, which is missing in Theorem~\ref{thm:known_wp}.
Analogous construction for CACs of length $p^r$ and its optimality are given as well.
\item  We relax the constant-weight constraint of traditional CACs to define mixed-weight CACs for the first time. We propose a general construction of mixed-weight CACs consisting of three or more different weights.
Finally, we provide three classes of optimal mixed-weight CACs containing two weights.
\end{enumerate}

The rest of this paper is organized as follows.
We set up some notations and useful results in Additive Combinatorics in Section~\ref{sec:Kneser}.
A new class of optimal CACs based on a direct construction is provided in Section~\ref{sec:CAC_direct}, while three classes of optimal CACs based on extending smaller-length CACs are proposed in~Section~\ref{sec:CAC-recursive}.
Section~\ref{sec:MW-CAC} is devoted to derive optimal mixed-weight CACs.
Some concluding remarks are given in Section~\ref{sec:conclusion}.

\section{Additive Combinatorics and Kneser's Theorem}\label{sec:Kneser}

A non-empty subset $S\subseteq\mathbb{Z}_L$ is said to be \textit{equi-difference} with \textit{generator} $g\in\mathbb{Z}_L$ if it is in the form $$\left\{0,g,2g,\ldots,(|S|-1)g\right\}.$$
Obviously, $|d^*(S)|\leq 2(|S|-1)$ when $S$ is equi-difference.
$S$ is called \textit{exceptional} if $|d^*(S)|<2(|S|-1)$.
Observe that in a CAC of length $L$, the union of all difference sets is a subset of $\mathbb{Z}^*_L$.
Therefore, if any $w$-subset of $\mathbb{Z}_L$ is not exceptional, then $K(L,w)\leq\lfloor\frac{L-1}{2w-2}\rfloor$.
So it is desired to characterize exceptional subsets in more details.

We need some results on Additive Combinatorics~\cite{TV06}.
For two subsets $A,B\subseteq\mathbb{Z}_L$ and an element $x\in\mathbb{Z}_L$, define
\begin{align*}
x+A &\triangleq \{x+a:\,a\in A\}, \\
A+B &\triangleq \{a+b:\,a\in A, b\in B\}, \text{ and} \\
A-B &\triangleq \{a-b:\,a\in A, b\in B\}.
\end{align*}
Moreover, define
\begin{align*}
d(A)\triangleq A-A.
\end{align*}
Note that $0\in d(A)$ and $d(A)\setminus\{0\}=d^*(A)$, the set of differences of $A$ given in~\eqref{eq:difference-set-def}.

Let $T$ be a non-empty subset in $\mathbb{Z}_L$.
The set of \textit{stabilizers} of $T$ in $\mathbb{Z}_L$ is defined as
\begin{align*}
\HH(T) \triangleq \{h\in\mathbb{Z}_L:\,h+T=T\}.
\end{align*}
It is obvious that $0\in\HH(T)$ and $\HH(T)$ is a subgroup of $\mathbb{Z}_L$.
So, it holds that $|\HH(T)|$ divides $L$ by Lagrange's theorem.
$T$ is said to be \textit{periodic} if $\HH(T)$ is non-trivial, that is, $\HH(T)\neq\{0\}$.
Here are some well-known (e.g., see~\cite{SWC10}) properties of the set of stabilizers. 

\begin{proposition}\label{prop:stabilizer}
Let $T\subseteq\mathbb{Z}_L$ be non-empty.
\begin{enumerate}[(i)]
\item $\HH(T)$ is a subgroup of $\mathbb{Z}_L$, and thus $|\HH(T)|$ divides $L$.
\item If $0\in T$, then $\HH(T)\subseteq T$.
\end{enumerate}
\end{proposition}

For example, consider the set $T=\{0,1,4,5,6,9\}\subseteq\Z_{10}$.
We have $\HH(T)=\{0,5\}$, which is a subgroup of $\Z_{10}$.
Note that $T=d(S)$, where $S=\{0,1,5,6\}$ is an exceptional subset in $\Z_{10}$.
One can see that $S+H=\{0,5\}\uplus\{1,6\}=H\uplus(1+H)$, a disjoint union of cosets of $H$.

In what follows, we will discuss $\HH(d(S))$ in more detail.
Before that, let us revisit Kneser's theorem~\cite[Theorem 5.5]{TV06}, which plays an important role in the derivation of the upper bound on the maximum number of codewords in a CAC.

\begin{theorem}[\cite{Kneser53,TV06}]\label{thm:Kneser}
Let $A$ and $B$ be two non-empty subsets in $\mathbb{Z}_L$, and let $H=\HH(A+B)$.
Then, 
\begin{equation}\label{eq:Kneser1}
|A+B|\geq |A+H|+|B+H|-|H|.
\end{equation}
In particular, 
\begin{equation}\label{eq:Kneser2}
|A+B|\geq |A|+|B|-|H|.
\end{equation}
\end{theorem}

By applying Theorem~\ref{thm:Kneser}, we immediately have the following corollary.

\begin{corollary}\label{cor:excpetional-bound}
Let $S$ be a $w$-subset in $\mathbb{Z}_L$.
If $S$ is exceptional, i.e., $|d^*(S)|<2w-2$, then $2\leq |\HH(d(S))|\leq 2w-2$.
\end{corollary}
\begin{proof}
Firstly, by definition, $d(S)=d^*(S)\uplus\{0\}$.
It follows that $|d(S)|=|d^*(S)|+1\leq 2w-2$.
Since $0\in d(S)$, by Proposition~\ref{prop:stabilizer}(ii), $|\HH(d(S))|\leq |d(S)|\leq 2w-2$.

Secondly, since $d(S)=S-S$, by plugging $A=S$ and $B=-S$ into~\eqref{eq:Kneser2}, we have
\begin{align*}
2w-2 & \geq |d(S)| = |S+(-S)| \\
& \geq |S|+|-S|-|\HH(d(S))| =2w-|\HH(d(S))|,
\end{align*}
which implies that $|\HH(d(S))|\geq 2$.
\end{proof}

We find some codewords will not be exceptional if the size of stabilizers of their difference sets satisfies some certain divisibility relations.

\begin{lemma}\label{lem:nonexceptional}
Let $L,w$ be positive integers.
For any $w$-subset $S\subseteq\Z_L$, if $|\mathsf{H}(d(S))|$ divides $w-1$ or $2w-1$, then $S$ is not exceptional.
\end{lemma}
\begin{proof}
Suppose to the contrary that $S$ is exceptional.
For notational convenience, denote by $H_S=\mathsf{H}(d(S))$.

We first consider the case when $|H_S|$ divides $w-1$.
Assume $w-1=k|H_S|$ for some integer $k\geq 1$.
Since $S$ is exceptional, we have 
\begin{equation}\label{eq:nonexcpetional_w-1_1}
|d(S)|=|d^*(S)|+1\leq 2w-2=2k|H_S|.
\end{equation}
Since $H_S$ is a subgroup of $\mathbb{Z}_{L}$, we have $H_S=-H_S$, which implies that $|-S+H_S|=|-(S+H_S)|=|S+H_S|$.
Plugging $A=S$ and $B=-S$ into \eqref{eq:Kneser1} yields that
\begin{align}
|d(S)| = |S+(-S)| &\geq |S+H_S| + |-S+H_S| - |H_S| \notag \\
&= 2|S+H_S| - |H_S|. \label{eq:nonexcpetional_w-1_2}
\end{align}
As $S+H_S$ is a disjoint union of cosets of $H_S$, $|H_S|$ divides $|S+H_S|$.
On the other hand, $|S+H_S|\geq |S|=w=k|H_S|+1$.
Hence we have $|S+H_S|\geq (k+1)|H_S|$.
It follows from \eqref{eq:nonexcpetional_w-1_1} and \eqref{eq:nonexcpetional_w-1_2} that
\begin{align*}
2k|H_S| \geq |d(S)|\geq 2|S+H_S|-|H_S| \geq (2k+1)|H_S|,
\end{align*}
which is a contradiction.

Now, consider the case when $|H_S|$ divides $2w-1$.
Assume $2w-1=h|H_S|$, for some odd $h\geq 1$.
Since $S$ is exceptional, we have
\begin{equation}\label{eq:nonexceptional_2w-1_1}
|d(S)|=|d^*(S)|+1 \leq 2w-2 = h|H|-1.
\end{equation}
Observe that $|S+H_S|\geq|S|=w=\frac{1}{2}(h|H_S|+1)>\frac{h}{2}|H_S|$.
Since $|H_S|$ divides $|S+H_S|$ and $h$ is odd, we further have $|S+H_S|\geq\frac{h+1}{2}|H_S|$.
Following the same argument in the derivation of \eqref{eq:nonexcpetional_w-1_2}, we have
\begin{equation}\label{eq:nonexceptional_2w-1_2}
|d(S)| \geq 2|S+H_S|-|H_S| \geq h|H_S|.
\end{equation}
It follows from \eqref{eq:nonexceptional_2w-1_1} and \eqref{eq:nonexceptional_2w-1_2} that $h|H_S|-1\geq |d(S)| \geq h|H_S|$, a contradiction occurs.
\end{proof}

\begin{example}\rm
Let $L=60$ and $w=4$.
Consider the following codewords $S_1,\ldots,S_5$ and the corresponding $d(S_i)$ and $\HH(d(S_i))$ for each $i$.
\begin{center}
\begin{tabular}{c|l|l|l}
$i$ & $S_i$ & $d(S_i)$ & $\HH(d(S_i))$ \\ \hline
$1$ & $\{0,15,30,45\}$ & $\{0,15,30,45\}$ & $\{0,15,30,45\}$ \\
$2$ & $\{0,12,24,36\}$ & $\{0,12,24,36,48\}$ & $\{0,12,24,36,48\}$ \\
$3$ & $\{0,10,20,30\}$ & $\{0,10,20,30,40,50\}$ & $\{0,10,20,30,40,50\}$ \\
$4$ & $\{0,8,30,38\}$ & $\{0,8,22,30,38,52\}$ & $\{0,30\}$ \\
$5$ & $\{0,8,16,24\}$ & $\{0,8,16,24,36,44,52\}$ & $\{0\}$
\end{tabular}
\end{center}
The first four codewords are exceptional and have $2\leq |\HH(d(S_i))|\leq 2w-2=6$, while the last one is non-exceptional and has $|\HH(d(S_5))|=1$.
This matches the assertions of Corollary~\ref{cor:excpetional-bound} and Lemma~\ref{lem:nonexceptional}.
Note that one can check that each $\HH(d(S_i))$ is a subgroup of $\Z_{60}$.
Note also that there is no codeword $S$ with $|\HH(d(S))|=3$ due to Lemma~\ref{lem:nonexceptional}.
So, this example covers all possible values $|\HH(d(S))|=2,4,5,6$ when $S$ is exceptional.
\end{example}

We recall a fundamental result in Group Theory, which will be used in subsequent sections.

\begin{proposition}\label{prop:unique_subgroup}
The subgroup of $\mathbb{Z}_L$ is uniquely determined by its order.
More precisely, for any divisor $d$ of $L$, the unique subgroup of $\mathbb{Z}_L$ with order $d$ is $\{0,L/d,2L/d,\ldots,(d-1)L/d\}$.
\end{proposition}

Finally, we summarize the useful notations throughout this paper in Table~\ref{tab:notation}.
Note that some notations will be defined in subsequent sections.

\begin{table}[ht]
\begin{center}
\begin{tabular}{|c|c|}
\hline
$\CAC(L,w)$ & The class of all CACs of length $L$ with weight $w$ \\ \hline
$\CACe(L,w)$ & The class of all equi-difference CACs of length $L$ with weight $w$ \\ \hline
$\CAC(L,\mathcal{W})$ & The class of all mixed-weight CACs of length $L$ with weight-set $\mathcal{W}$ \\ \hline
$K(L,w)$ & The maximum size of a code in $\CAC(L,w)$ \\ \hline
$K^e(L,w)$ & The maximum size of a code in $\CACe(L,w)$ \\ \hline
\multirow{2}{2cm}{$K(L,w;w^*,n)$} 
 & The maximum size of a code in $\CAC(L,\{w,w^*\})$ where \\
 & the number of codewords with weight $w^*$ is exactly $n$ \\ \hline
$\Z_L$ & $\{0,1,2,\ldots,L-1\}$, the ring of residue modulo $L$ \\ \hline
$\Z^*_L$ & The set of nonzeros in $\Z_L$ \\ \hline
$\Z^\times_L$ & The set of units in $\Z^*_L$ \\ \hline
$\theta$ & The CRT correspondence \\ \hline
$d(S)$ & $S-S$, the set of differences of $S$ \\ \hline
$d^*(S)$ & $d(S)\setminus\{0\}$, the set of nonzero differences of $S$ \\ \hline
$\HH(T)$ & The set of stabilizers of $T$ \\ \hline
$\left(\frac{a}{p}\right)$ & The Legendre symbol of $a$ on $\Z_p$ \\ \hline
\multirow{2}{2mm}{$\mathsf{L}_t$} 
 & The $t$-th layer in $\Z^*_{p^r}$, consisting of elements whose nonzero \\
 & least significant digit in their $p$-representation is $p^t$  \\ \hline
\multirow{2}{7mm}{$\mathcal{S}_r(A)$} 
 & The set of elements in $\Z^*_{p^r}$ whose nonzero least significant \\
 & digit values in their $p$-ary representation are in $A$ \\ \hline
\multirow{2}{8mm}{$H^e(p)$} 
 & $\langle\alpha^e\rangle$, the multiplicative subgroup of $\Z^\times_p$ generated by $\alpha^e$, \\
 & where $\alpha$ is a primitive element \\ \hline
$H^e_j(p)$ & The coset $\alpha^j H^e(p)$ \\ \hline
$\mathcal{H}^e(p)$ & The collection of cosets of $H^e(p)$ \\ \hline
\end{tabular}
\end{center}
\caption{Main notations used in this paper}
\label{tab:notation}
\end{table}

\section{New Optimal CACs based on Direct Constructions}\label{sec:CAC_direct}

\subsection{$p$-ary representation}\label{sec:CAC_direct_p}

We first introduce the $p$-ary representation of a positive integer and its useful properties.

Given a positive integer $n$, let 
\begin{align*}
\mathbb{Z}^{\times}_n\triangleq\{x\in\mathbb{Z}_n:\,\gcd(x,n)=1\}.
\end{align*}
$\mathbb{Z}^{\times}_n$ is the set of \emph{units} (i.e., invertible elements) in $\mathbb{Z}^*_n$, and thus is a multiplicative group.
Note that $\mathbb{Z}^{\times}_n=\mathbb{Z}^{*}_n$ when $n$ is a prime.

Let $p$ be an odd prime and $r$ a positive integer.
For $c\in\mathbb{Z}_{p^r}$, consider the $p$-ary representation $c=c_0+c_1p+\cdots+c_{r-1}p^{r-1}$.
For $t=0,1,\ldots,r-1$, let $\mathsf{L}_t$ be the collection of $c\in\mathbb{Z}^{*}_{p^r}$ whose nonzero least significant digit in its $p$-ary representation is $p^t$.
Obviously, $|\mathsf{L}_t|=(p-1)p^{r-t-1}$, and $\mathsf{L_0},\mathsf{L}_1,\ldots,\mathsf{L}_{r-1}$ form a partition of $\mathbb{Z}^{*}_{p^r}$, i.e., $\mathbb{Z}_{p^r}^{*}=\mathsf{L}_0\uplus\mathsf{L}_1\uplus\cdots\uplus\mathsf{L_{r-1}}$.
Integers in $\mathsf{L}_t$ are called in the $t$-th \textit{layer}.

For a non-empty $A\subseteq\mathbb{Z}^{*}_p$, we arise it to a subset in $\mathbb{Z}^{*}_{p^r}$, for any positive integer $r$, by defining
\begin{equation}\label{eq:S_r(A)-def}
\mathcal{S}_r(A) \triangleq A_0\uplus A_1\uplus \cdots \uplus A_{r-1},
\end{equation}
where 
\begin{equation*}
A_t = \{c\in\mathsf{L}_t:\,c_t\in A\}.
\end{equation*}
$\mathcal{S}_r(A)$ is the collection of elements in $\mathbb{Z}^{*}_{p^r}$ whose nonzero least significant digit values in their $p$-ary representation are in $A$.
Obviously, $|A_t|=|A|p^{r-1-t}$ for each $t$, and thus 
\begin{equation}\label{eq:S_r(A)-size}
|\mathcal{S}_r(A)|=|A|\left(1+p+\cdots+p^{r-1}\right) =|A|\frac{p^r-1}{p-1}.
\end{equation}

Here is a useful property of the $p$-ary representation of $c\in\mathbb{Z}^{*}_{p^r}$, where the proof is straightforward and is omitted.

\begin{proposition}\label{prop:p-ary}
Let $p$ be an odd prime and $r$ a positive integer.
For $j\in\mathsf{L}_0$ and $c\in\mathsf{L}_t$, $0\leq t\leq r-1$, one has $jc\in\mathsf{L}_t$ and
\begin{equation}\label{eq:p-ary}
(jc)_t = j_0\cdot c_t \ (\bmod\, p).
\end{equation}
Note that $\mathsf{L}_0=\mathbb{Z}^{\times}_{p^r}$, which is the set of units in $\mathbb{Z}^*_{p^r}$.
\end{proposition}

\subsection{A direct construction}


Let $p$ be a prime and $\alpha\in\Z_p$ be a primitive element, i.e., $\Z^{\times}_p=\langle\alpha\rangle\triangleq\{\alpha^i:\,0\leq i\leq p-2\}$.
For any divisor $e$ of $p-1$, let $H^e(p)\triangleq\langle\alpha^e\rangle$ denote the multiplicative subgroup of $\Z_p^\times$ generated by $\alpha^e$, and denote by $$\mathcal{H}^e(p)\triangleq\{H^e_j(p)=\alpha^j\langle\alpha^e\rangle:\,j=0,1,\ldots,e-1\}$$ the collection of cosets of $H^e(p)$.
A set $\{a_0,a_1,\ldots,a_{e-1}\}$ of $e$ distinct elements in $\Z^{\times}_p$ is called a \emph{system of distinct representative (SDR)} of $\mathcal{H}^e(p)$ if $a_j\in H^e_j(p)$ for $0\leq j<e$.


\begin{theorem}\label{thm:construction_direct}
Let $w,d$ be positive integers and $p$ be a prime such that $d|(w-1)$, $2d|(p-1)$ and $p\geq w$.
If 
\begin{equation}\label{eq:SDR1}
\left\{\pm 1,\pm 2,\ldots, \pm d\right\} \text{ forms an SDR of }\mathcal{H}^{2d}(p),
\end{equation}
and for $1\leq i\leq \frac{w-1}{d}-1$, 
\begin{equation}\label{eq:SDR2}
\left\{ i+\frac{j(w-1)}{d}, i-\frac{(j+1)(w-1)}{d}:\,j=0,1,\ldots,d-1 \right\} \text{ forms an SDR of }\mathcal{H}^{2d}(p),
\end{equation}
then for any integer $r\geq 1$, there exists a code $\C\in\CACe(\frac{w-1}{d}p^r,w)$ with $(p^r-1)/2d$ codewords.
\end{theorem}
\begin{proof}
Since $\gcd(\frac{w-1}{d},p)=1$ due to $w\leq p$, one has $\mathbb{Z}_{\frac{w-1}{d}p^r}\cong\mathbb{Z}_{\frac{w-1}{d}}\times\mathbb{Z}_{p^r}$.
So, for the sake of convenience, the elements of codewords are represented as ordered pairs in $\mathbb{Z}_{\frac{w-1}{d}}\times\mathbb{Z}_{p^r}$ due to the CRT correspondence as shown in \eqref{eq:CRT-correspondence}.

Suppose $\alpha$ is a primitive element of $\mathbb{Z}^{\times}_p$.
Let $\Gamma=\{\alpha^{2dj}:\,0\leq j<\frac{p-1}{2d}\}$.
In other words, $\mathbb{Z}^{\times}_p=\langle\alpha\rangle$ and $\Gamma=H^{2d}(p)\,\,(=\langle\alpha^{2d}\rangle)$.
For $g\in\mathcal{S}_r(\Gamma)$, define a $w$-subset 
\begin{align*}
S_g\triangleq\{k(1,g)\in\mathbb{Z}_{(w-1)/d}\times\mathbb{Z}_{p^r}:\,k=0,1,2,\ldots,w-1\}.
\end{align*}
Note that $|\mathcal{S}_r(\Gamma)|=(p^r-1)/2d$ by \eqref{eq:S_r(A)-size}.
We claim that $\{S_g:\,g\in\mathcal{S}_r(\Gamma)\}$ forms the desired code, that is, $d^*(S_g)$, $g\in\mathcal{S}_r(\Gamma)$ are mutually disjoint.

The difference set of $S_g$ can be written as
\begin{align*}
d^*(S_g)&=\{k(1,g)\in\mathbb{Z}_{\frac{w-1}{d}}\times\mathbb{Z}_{p^r}:\,k=\pm1,\pm2,\ldots,\pm(w-1)\} \\
&=\{k(1,g)\in\mathbb{Z}_{\frac{w-1}{d}}\times\mathbb{Z}_{p^r}:\,k\in T_0\uplus T_1\uplus\cdots\uplus T_{\frac{w-1}{d}-1}\},
\end{align*}
where 
\begin{equation*}
T_0\triangleq\{j\cdot\frac{w-1}{d}:\,j=\pm1,\pm2,\ldots,\pm d\}
\end{equation*}
and
\begin{equation*}
T_s\triangleq\{s+\frac{j(w-1)}{d}, s-\frac{(j+1)(w-1)}{d}:\,j=0,1,\ldots,d-1\},
\end{equation*}
for $1\leq s\leq \frac{w-1}{d}-1$.
Note that, if $k\in\{\pm1,\pm2,\ldots,\pm(w-1)\}$ such that $k=s$ (mod $\frac{w-1}{d}$) for some $0\leq s\leq\frac{w-1}{d}-1$, then $k\in T_s$.
By the assumptions \eqref{eq:SDR1} and \eqref{eq:SDR2}, one has
\begin{equation}\label{eq:SDR3}
T_s \text{ forms an SDR of } \mathcal{H}^{2d}(p), \, \forall\, 0\leq s\leq \frac{w-1}{d}-1.
\end{equation}

Suppose to the contrary that $d^*(S_g)\cap d^*(S_h)\neq\emptyset$ for some $g\neq h$, $g,h\in\mathcal{S}_r(\Gamma)$.
Assume 
\begin{equation}\label{eq:assume-contrary}
i(1,g)=j(1,h)\in\Z_{\frac{w-1}{d}}\times\Z_{p^r}    
\end{equation}
for some $i,j\in\{\pm1,\pm2,\ldots,\pm(w-1)\}$.
The first component in~\eqref{eq:assume-contrary} implies that $i,j\in T_s$ for some $s$, while the second component indicates that $ig$ and $jh$ are on the same layer, say $\mathsf{L}_t$.
Since $i,j\in\{\pm1,\ldots,\pm(w-1)\}\subseteq\mathsf{L}_0$ due to $w\leq p$, by Proposition~\ref{prop:p-ary}, we have $g,h\in\mathsf{L}_t$.
More precisely, $g=i^{-1}\cdot ig\in\mathsf{L}_t$ and $h=j^{-1}\cdot jh\in\mathsf{L}_t$.
Therefore,
\begin{equation}\label{eq:assume-contrary-2}
i\cdot g_t=j\cdot h_t \ (\bmod\ p).
\end{equation}
Observe that $g_t,h_t\in H^{2d}(p)$ and $i,j\in T_s$, which is an SDR of $\mathcal{H}^{2d}(p)$ by~\eqref{eq:SDR3}.
It follows from~\eqref{eq:assume-contrary-2} that $i=j$.
Finally, plugging $i=j$ into the second component in~\eqref{eq:assume-contrary} yields $i(g-h)=0$, implying $g=h$, which is a contradiction.
\end{proof}

\begin{example}\rm
Let $p=37, w=7$ and $d=2$.
We have $(w-1)/d=3$.
The $2d=4$ cosets of $\mathcal{H}^{4}(37)$ are
\begin{align*}
H^4_0(37) &= \{1,7,9,10,12,16,26,33,34\}, \\
H^4_1(37) &= \{5,6,8,13,17,19,22,23,35\}, \\
H^4_2(37) &= \{3,4,11,21,25,27,28,30,36\}, \\
H^4_3(37) &= \{2,14,15,18,20,24,29,31,32\}.
\end{align*}
One can verify that each of $\{\pm1,\pm2\}$, $\{1,-2,4,-5\}$, $\{-1,2,-4,5\}$ forms an SDR of $\mathcal{H}^{4}(37)$.
By Theorem~\ref{thm:construction_direct}, we have an equi-difference CAC $\C$ of length $3\cdot {37}^r$ and weight $7$ with $|\C|=(37^r-1)/4$, for any integer $r\geq 1$.
When $r=1$, the set of generators is $\{\theta^{-1}((1,g)):\,g\in H_0^4(37)\}=\{1,7,10,16,34,46,49,70,100\}$, where $\theta:\mathbb{Z}_{111}\to\mathbb{Z}_3\times\mathbb{Z}_{37}$ is the bijection given in~\eqref{eq:CRT-correspondence}.
When $r=2$, the set of generators of $\C$ is
\begin{align*}
\left\{\theta^{-1}((1,a+37b)):\,a\in H^4_0(37)\text{ and }0\leq b\leq 36 \text{ or } a=0 \text{ and }b\in H^4_0(37)\right\},
\end{align*} 
where $\theta$ is the bijective mapping $\mathbb{Z}_{4107}\to\mathbb{Z}_3\times\mathbb{Z}_{37^2}$ now.
The corresponding generators for $b=0$ are: $1,7,10,16,34, 1378, 1381, 1402, 2764$, and for $a=0$ are: $37, 259, 370, 592, 1258, 1702, 1813$, $2590$, $3700$.
\end{example}

\subsection{Optimal CACs of length $\frac{w-1}{d}p^r$ and weight $w$}

Here is our main result in this section, which generalizes both \cite[Theorem 3.7]{MMSJ07} and \cite[Theorem 7]{SW10}.

\begin{theorem}\label{thm:main_w-1dpr}
Let $w,d$ be positive integers and $p$ be a prime such that $d|(w-1)$, $2d|(p-1)$ and $p\geq w$.
If the two conditions in \eqref{eq:SDR1} and \eqref{eq:SDR2} hold, then for any integer $r\geq 1$,
\begin{align*}
K\left(\frac{w-1}{d}p^r,w\right)=\frac{p^r-1}{2d}.
\end{align*}
\end{theorem}
\begin{proof}
By Theorem~\ref{thm:construction_direct}, it suffices to show that for any code $\mathcal{C}\in\CAC(\frac{w-1}{d}p^r,w)$, one has $|\mathcal{C}| \leq \frac{p^r-1}{2d}$.

Let $\mathcal{E}\subseteq\mathcal{C}$ be the collection of all exceptional codewords in $\C$.
For notational convenience, denote by $H_S=\HH(d(S))$ for $S\in\mathcal{E}$.
Since, by Proposition~\ref{prop:stabilizer}(i), $H_S$ must contain the element $0$, so we further denote by $H_S^*=H_S\setminus\{0\}$.

Consider any $S\in\mathcal{E}$.
Since $H_S$ is a subgroup of $\mathbb{Z}_{\frac{w-1}{d}p^r}$, one has $|H_S|$ divides $\frac{w-1}{d}p^r$.
Moreover, $|H_S|$ does not divide $\frac{w-1}{d}$ since it does not divide $w-1$ by Lemma~\ref{lem:nonexceptional}.
As $|H_S|\geq 2$ by Corollary~\ref{cor:excpetional-bound}, it follows that $|H_S|$ is a multiple of $p$.
Note that $|H_S|\leq 2w-2$ by Corollary~\ref{cor:excpetional-bound} again.
We consider two cases.

\textit{Case 1:} $p>2w-2$. 
In this case we have $|H_S|\leq 2w-2<p$, which is a contradiction to $|H_S|$ a multiple of $p$.
In other words, there is no exceptional codeword in this case.

\textit{Case 2:} $p\leq 2w-2$.
Since $|H_S|\leq 2w-2\leq 2p-2$, it must be the case that $|H_S|=p$.
By Proposition~\ref{prop:unique_subgroup}, such an exceptional codeword is unique.
As $0\in d(S)$, Proposition~\ref{prop:stabilizer}(ii) implies that $|H_{S}|\leq|d(S)|=|d^*(S)|+1$.
So we have $d^*(S)\geq |H_S|-1=p-1$.

It concludes that there is at most one exceptional codeword in $\C$, and the unique codeword, denoted by $\widehat{S}$ if exists, satisfies $d^*(\widehat{S})\geq p-1$.
When $\widehat{S}$ does not exist, by the disjoint-difference-set property, we have
\begin{align*}
\frac{w-1}{d}p^r-1 = |\mathbb{Z}^*_{\frac{w-1}{d}p^r}| \geq (2w-2)|\C|,
\end{align*}
and then 
\begin{align*}
|\C|\leq \left\lfloor\frac{p^r-1}{2d}+\frac{\frac{w-1}{d}-1}{2w-2}\right\rfloor = \frac{p^r-1}{2d}.
\end{align*}
When $\widehat{S}$ does exist, by the disjoint-difference-set property, we have
\begin{align*}
\frac{w-1}{d}p^r-1 = |\mathbb{Z}^*_{\frac{w-1}{d}p^r}| \geq (2w-2)(|\C|-1) + (p-1),
\end{align*}
and then
\begin{align*}
|\C| \leq \left\lfloor\frac{p^r-1}{2d} + 1-\frac{p-\frac{w-1}{d}}{2w-2}\right\rfloor = \frac{p^r-1}{2d},
\end{align*}
where the last equality is due to the necessary condition $p\leq 2w-2$ of the existence of $\widehat{S}$.
\end{proof}

In the rest of this section, we will obtain a series of optimal CACs by exploring primes $p$ that satisfy the two conditions in \eqref{eq:SDR1} and \eqref{eq:SDR2}.
Note that when $d=1$, the subgroup $H^{2}(p)$ is the same as $Q(p)$, the group consists of all quadratic residues modulo $p$.
The two conditions \eqref{eq:SDR1} and \eqref{eq:SDR2} are then identical to \eqref{eq:QR1} and \eqref{eq:QR2}, respectively, and hence Theorem~\ref{thm:main_w-1dpr} (for the case of $r=1$) can be reduced to Theorem~\ref{thm:known_w-1p}.

%
%

We first list some well-known results in the followings (e.g., see \cite[Theorems 9.6, 9.10, and Problem~10 in Chapter~9.3]{Burton10}). 
Note that these results can be derived by Gauss's Lemma (e.g., \cite[Theorem 9.5]{Burton10}) and the Law of Quadratic Reciprocity.

\begin{lemma}[\cite{Burton10}]\label{lem:QRLaw}
Let $p$ be an odd prime. One has
\begin{enumerate}[(i)]
\item $\big(\frac{-1}{p}\big)=-1$ if and only if $p\equiv 3\ (\bmod\ 4)$,
\item $\big(\frac{2}{p}\big)=1$ if and only if $p\equiv \pm 1\ (\bmod\ 8)$,
\item $\big(\frac{3}{p}\big)=1$ if and only if $p\equiv \pm 1\ (\bmod\ 12)$,
\item $\big(\frac{5}{p}\big)=1$ if and only if $p\equiv \pm 1\ (\bmod\ 10)$,
\item $\big(\frac{6}{p}\big)=1$ if and only if $p\equiv \pm 1,\pm 5\ (\bmod\ 24)$, and
\item $\big(\frac{7}{p}\big)=1$ if and only if $p\equiv \pm 1,\pm 3,\pm 9\ (\bmod\ 28)$.
\end{enumerate}
\end{lemma}

We have the following optimal CACs.

\begin{corollary}\label{coro:w-1prw}
Let $p$ be an odd prime and $r$ be any positive integer.
One has 
\begin{enumerate}[(i)]
\item $K(3p^r,4)=(p^r-1)/2$ if $p\equiv -1\ (\bmod\ 8)$, 
\item $K(4p^r,5)=(p^r-1)/2$ if $p\equiv -1\ (\bmod\ 12)$, 
\item $K(5p^r,6)=(p^r-1)/2$ if $p\equiv -1,-5\ (\bmod\ 24)$, 
\item $K(6p^r,7)=(p^r-1)/2$ if $p\equiv -1,-9\ (\bmod\ 40)$, 
\item $K(7p^r,8)=(p^r-1)/2$ if $p\equiv -1,-49\ (\bmod\ 120)$, 
\item $K(8p^r,9)=(p^r-1)/2$ if  $p\equiv -1,59,-109,-121,131,-169\ (\bmod\ 420)$, 
\item $K(9p^r,10)=(p^r-1)/2$ if  $p\equiv -1,-9,31,-81,111,-121\ (\bmod\ 280)$, and
\item $K(10p^r,11)=(p^r-1)/2$ if  $p\equiv -1,-5,-25,43,47,67\ (\bmod\ 168)$.
\end{enumerate}
\end{corollary}
\begin{proof}
It is routine to simplify the two conditions in \eqref{eq:QR1} and \eqref{eq:QR2} in a system of quadratic-residue equations, as shown in the following table.
For example, when $w=7$, \eqref{eq:QR2} implies $\big(\frac{1}{p}\big)\big(\frac{-5}{p}\big)=\big(\frac{2}{p}\big)\big(\frac{-4}{p}\big)=\big(\frac{3}{p}\big)\big(\frac{-3}{p}\big)=-1$.
Since $\big(\frac{-1}{p}\big)=-1$ by \eqref{eq:QR1}, the identity $\big(\frac{3}{p}\big)\big(\frac{-3}{p}\big)=-1$ automatically hold.
Meanwhile, $\big(\frac{1}{p}\big)\big(\frac{-5}{p}\big)=-1$ implies $\big(\frac{5}{p}\big)=1$ and $\big(\frac{2}{p}\big)\big(\frac{-4}{p}\big)=-1$ implies $\big(\frac{2}{p}\big)=1$.

\begin{center}
\begin{tabular}{c|l}
$w$ & simplified equations of \eqref{eq:QR1} and \eqref{eq:QR2} \\ \hline
$4$ & $\big(\frac{-1}{p}\big)=-1$ and $\big(\frac{2}{p}\big)=1$ \\
$5$ & $\big(\frac{-1}{p}\big)=-1$ and $\big(\frac{3}{p}\big)=1$ \\
$6$ & $\big(\frac{-1}{p}\big)=-1$ and $\big(\frac{6}{p}\big)=1$ \\
$7$ & $\big(\frac{-1}{p}\big)=-1$ and $\big(\frac{2}{p}\big)=\big(\frac{5}{p}\big)=1$ \\
$8$ & $\big(\frac{-1}{p}\big)=-1$ and $\big(\frac{2}{p}\big)=\big(\frac{3}{p}\big)=\big(\frac{5}{p}\big)=1$ \\
$9$ & $\big(\frac{-1}{p}\big)=-1$ and $\big(\frac{3}{p}\big)=\big(\frac{5}{p}\big)=\big(\frac{7}{p}\big)=1$ \\
$10$ & $\big(\frac{-1}{p}\big)=-1$ and $\big(\frac{2}{p}\big)=\big(\frac{5}{p}\big)=\big(\frac{7}{p}\big)=1$ \\
$11$ & $\big(\frac{-1}{p}\big)=-1$ and $\big(\frac{6}{p}\big)=\big(\frac{3}{p}\big)\big(\frac{7}{p}\big)=1$ \\
\end{tabular}
\end{center}
Then, each of above systems of equations can be solved by Lemma~\ref{lem:QRLaw}.
\end{proof}

\begin{remark}\rm
For any arbitrary $w$, we can derive a sufficient condition of primes $p$ so that $K((w-1)p^r,w)=(p^r-1)/2$ as long as the corresponding Quadratic Reciprocity Laws are obtained.
This is workable because the latter can be done by Gauss's Lemma.
\end{remark}

Now, let us turn to $d=2$.
It was shown in \cite[Corollary 3.10]{MMSJ07} that any prime $p\equiv 5\ (\bmod\ 24)$ satisfies the two conditions in \eqref{eq:SDR1} and \eqref{eq:SDR2} for the case when $w=5$ and $d=2$.
By Theorem~\ref{thm:main_w-1dpr}, we immediately have the following result.

\begin{corollary}\label{coro:2pr5}
Let $p\equiv 5\ (\bmod\ 24)$ be a prime.
Then, for any integer $r\geq 1$, one has $K(2p^r,5)=(p^r-1)/4$.
\end{corollary}

We further figure out one class of primes that satisfies conditions in \eqref{eq:SDR1} and \eqref{eq:SDR2} for the case  when $w=7$ and $d=2$.

\begin{corollary}\label{coro:3pr7}
Let $p=5\ (\bmod\ 8)$ be a prime with $10^{(p-1)/4}\equiv 1\ (\bmod\ p)$.
Then, for any integer $r\geq 1$, one has $K(3p^r,7)=(p^r-1)/4$.
\end{corollary}
\begin{proof}
The two conditions in \eqref{eq:SDR1} and \eqref{eq:SDR2} claim that each of $\{1,-1,2,-2\}$, $\{1,-2,4,-5\}$, $\{-1,2,-4,5\}$ forms an SDR of $\mathcal{H}^4(p)=\{H^4_0(p),H^4_1(p),H^4_2(p),H^4_3(p)\}$.
As, in the first set, $1$ and $-1$ are in distinct cosets, $\{1,-2,4,-5\}$ is an SDR if and only if $\{-1,2,-4,5\}$ is an SDR.
So, it suffices to consider the first two sets $\{1,-1,2,-2\}$ and $\{1,-2,4,-5\}$.
Note that $H^4_0(p)\cup H^4_2(p)=Q(p)$, the collection of quadratic residues modulo $p$.

Suppose $\alpha$ is a primitive element of $\mathbb{Z}_p^{\times}$.
Observe that an element $\alpha^e\in H^4_i(p)$ if and only if $e\equiv i\ (\bmod\ 4)$.
Since $-1=\alpha^{(p-1)/2}$ and $(p-1)/2\equiv 2\ (\bmod\ 4)$, we have $-1\in H^4_2(p)$.
By Lemma~\ref{lem:QRLaw}(ii), $2\notin Q(p)$.
So, either $2\in H^4_1(p)$ or $2\in H^4_3(p)$.
As $-1\in H^4_2(p)$, we further have either $2\in H^4_1(p)$ and $-2\in H^4_3(p)$ or $2\in H^4_3(p)$ and $-2\in H^4_1(p)$.
Hence $\{1,-1,2,-2\}$ is an SDR.

Now, consider the set $\{1,-2,4,-5\}$.
The assumption $10^{(p-1)/4}\equiv 1\ (\bmod\ p)$ makes sure that $10\in H^4_0(p)$.
Since $-1\in H^4_2(p)$ and $2$ is either in the coset $H^4_1(p)$ or $H^4_3(p)$, we have either $-2\in H^4_1(p)$ and $-5\in H^4_3(p)$ or $-2\in H^4_3(p)$ and $-5\in H^4_1(p)$.
This completes the proof.
\end{proof}

The primes that satisfy the conditions given in Corollary~\ref{coro:3pr7} are 37, 53, 173, 277, 317, 397, 613, 733, 757, 773, 797, and so on.

\section{New Optimal CACs based on Extending Smaller-length CACs}\label{sec:CAC-recursive}

Let $p$ be a prime.
This section includes three classes of CACs of length $L=p^r, wp^r$, and $(2w-1)p^r$, by extending a code in $\CAC(p,w)$.

\subsection{Optimal CACs of length $p^r$}

This subsection is devoted to generalize the following result, given in \cite[Theorem 6]{SW10}.

\begin{theorem}[\cite{SW10}] \label{thm:pr-recall-optimal}
For any odd prime $p$ and positive integer $r$, 
\begin{align*}
K\left(p^r,(p+1)/2\right)=\frac{p^r-1}{p-1}.
\end{align*}
\end{theorem}

A code in $\CAC(p^r,(p+1)/2)$ whose size attains the maximum code size is proposed in~\cite[Theorem 10]{Levenshtein07}.
Such a code is equi-difference with generators consist of all integers in $\Z_{p^r}$ whose first nonzero symbol in the $p$-ary representation is $1$.
Inspired by this construction, we have the following result.

\begin{theorem}\label{thm:construction_pr}
Let $p$ be a prime such that $p\geq 2w-1$.
If there is a code $\mathcal{C}\in\CACe(p,w)$ with $m$ codewords, then for any integer $r\geq 1$, there exists a code in $\CACe(p^r,w)$ with $m(p^r-1)/(p-1)$ codewords.
\end{theorem}
\begin{proof}
Let $\Gamma$ denote the set of $m$ generators of $\mathcal{C}$.
By definition, one has $ig\neq jh\ (\bmod\, p)$ for $i,j\in\{\pm1,\pm2,\ldots,\pm (w-1)\}$, $g,h\in\Gamma$ provided that $g\neq h$.


Consider the set $\mathcal{S}_r(\Gamma)$.
For $g\in\mathcal{S}_r(\Gamma)$, define a $w$-subset $S_g=\{jg\in\mathbb{Z}_{p^r}:\,j=0,1,2,\ldots,w-1\}$, whose difference set is of the form $d^*(S_g)=\{jg\in\mathbb{Z}_{p^r}:\,j=\pm1,\pm2,\ldots,\pm(w-1)\}$.
In what follows, we shall show that these $w$-subsets form a code in $\CACe(p^r,w)$, that is, $d^*(S_g)\cap d^*(S_h)=\emptyset$ for any distinct $g,h\in\mathcal{S}_r(\Gamma)$.
Hence, by \eqref{eq:S_r(A)-size}, there are $m(p^r-1)/(p-1)$ codewords, as desired.

Since $p\geq 2w-1$, by Proposition~\ref{prop:p-ary}, one has $d^*(S_g)\subset\mathsf{L}_t$ if $g\in\mathsf{L}_t$.
It follows that $d^*(S_g)\cap d^*(S_h)=\emptyset$ whenever $g$ and $h$ are in distinct layers in the $p$-ary representation.
Now, it suffices to consider the case when $g$ and $h$ are in the same layer, say $\mathsf{L}_t$ for some $t$.
Suppose to the contrary that $ig=jh$ (mod $p^r$) for some $i,j\in\{\pm1,\pm2,\cdots,\pm(w-1)\}$.
One has $i,j\in\mathsf{L}_0$ due to $p\geq 2w-1$.
By Proposition~\ref{prop:p-ary} again, $i\cdot g_t=j\cdot h_t$ (mod $p$).
If $g_t\neq h_t$, then a contradiction occurs due to the assumption that $g_t,h_t\in\Gamma$ are two distinct generators in the given code $\mathcal{C}\in\CACe(p,w)$.
If $g_t=h_t$, it further implies that $(i-j)g_t=0$ (mod $p$), which is impossible because of $i,j\in\{\pm1,\pm2,\cdots,\pm(w-1)\}$ and $p\geq 2w-1$.
This completes the proof.
\end{proof}


To prove $K(p^r,(p+1)/2)\leq (p^r-1)/(p-1)$ in Theorem~\ref{thm:pr-recall-optimal}, \cite{SW10} provided a more general result as follows.

\begin{theorem}[\cite{SW10}, Theorem 5]\label{thm:pr-recall-upper}
Suppose the prime factors of $L$ are all larger than or equal to $2w-1$, then 
\begin{align*}
K(L,w)\leq\left\lfloor\frac{L-1}{2w-1}\right\rfloor.
\end{align*}
\end{theorem}

Note that the proof of Theorem~\ref{thm:pr-recall-upper} is to exclude the existence of exceptional codewords. 
The key idea is that $\HH(d(S))$ is a non-trivial subgroup of $\Z_L$ whenever $S$ is exceptional, which can also be deduced by Proposition~\ref{prop:stabilizer} and Corollary~\ref{cor:excpetional-bound}.

We now show the construction described in Theorem~\ref{thm:construction_pr} is optimal in some cases.

\begin{theorem}\label{thm:main_pr}
Let $p$ be a prime such that $p-1$ is divided by $2w-2$.
If there is a code in $\CACe(p,w)$ with $(p-1)/(2w-2)$ codewords, then for any integer $r\geq 1$,
\begin{align*}
K\left(p^r,w\right)=\frac{p^r-1}{2w-2}.
\end{align*}
\end{theorem}
\begin{proof}
The assumption that $p-1$ is divisible by $2w-2$ guarantees $p\geq 2w-1$.
By Theorem~\ref{thm:construction_pr}, there exists a code in $\CACe(p^r,w)$ with $(p^r-1)/(2w-2)$ codewords.
As $p^r$ has only one prime factor $p$, the result follows from Theorem~\ref{thm:pr-recall-upper}.

\end{proof}

\begin{example}\label{ex:37r4} \rm
Let $p=37$ and $w=4$.
One can check that $\Gamma=\{1,6,8,10,11,14\}$ forms a set of generators of a code in $\CACe(37,4)$.
By Theorem~\ref{thm:main_pr}, we have $K(37^r,4)=(37^r-1)/6$ for any integer $r\geq 1$.
Take $r=2$ as an example.
The code in $\CAC^e(37^2,4)$ obtained by the construction of Theorem~\ref{thm:construction_pr} is of size $228$, in which the set of generators is
\begin{align*}
\{a+37b:\,a\in\Gamma\text{ and }0\leq b\leq 36\text{ or }a=0\text{ and }b\in\Gamma\}.
\end{align*} 
\end{example}


\subsection{Optimal CACs of length $wp^r$}

\begin{theorem}\label{thm:construction_wpr}
Let $p$ be a prime such that $p\geq 2w-1$.
If there is a code in $\CACe(p,w)$ with $m$ codewords and 
\begin{equation}\label{eq:quadratic-non-residue}
\left(\frac{i}{p}\right)\left(\frac{i-w}{p}\right)=-1,\ \forall i=1,2,\ldots,w-1, 
\end{equation}
then for any integer $r\geq 1$, there exists a code $\C\in\CACe(wp^r,w)$ with
\begin{align*}
|\C|=\frac{m(p^r-1)}{p-1}+\frac{p^r-1}{2}+1
\end{align*}  
codewords.
\end{theorem}
\begin{proof}
Let $\Gamma$ be a set of $m$ generators of the given code in $\CACe(p,w)$, and $Q=Q(p)$ be the set of quadratic residues modulo $p$.
Define the two sets 
\begin{align*}
\widehat{\Gamma}\triangleq\{(0,g)\in\mathbb{Z}_{w}\times\mathbb{Z}_{p^r}:g\in\mathcal{S}_r(\Gamma)\}
\end{align*}
and
\begin{align*}
\widehat{Q}\triangleq\{(1,g)\in\mathbb{Z}_{w}\times\mathbb{Z}_{p^r}:\,g\in\mathcal{S}_r(Q)\}.
\end{align*}
It is obvious that $\widehat{\Gamma}$ and $\widehat{Q}$ are disjoint.
We shall prove that $\widehat{\Gamma}\uplus\widehat{Q}\uplus\{(1,0)\}$ is the set of generators of the desired code $\C$.
As such, by~\eqref{eq:S_r(A)-size}, the number of obtained codewords will be $|\widehat{\Gamma}|+|\widehat{Q}|+1=\frac{m(p^r-1)}{p-1}+\frac{p^r-1}{2}+1$.

Since $p$ is a prime with $p\geq 2w-1$, we have $\gcd(w,p)=1$ and thus $\Z_{wp^r}\cong\Z_w\times\Z_{p^r}$.
For $a\in\widehat{\Gamma}\uplus\widehat{Q}\uplus\{(1,0)\}$ let $S_a=\{ja:\,j=0,1,\ldots,w-1\}$ be the $w$-subset generated by $a$.
We will show $d^*(S_a)\cap d^*(S_b)=\emptyset$ whenever $a\neq b$.
Obviously $\widehat{\Gamma}\cap\widehat{Q}=\emptyset$ and $(1,0)\notin\widehat{\Gamma}$, and $(1,0)\notin\widehat{Q}$ as $0\notin Q$, which guarantees $d^*(S_a)\cap d^*(S_b)=\emptyset$ when $a$ and $b$ are in different sets of $\widehat{\Gamma}$, $\widehat{Q}$, or $\{(1,0)\}$.
It also holds in the case when both $a,b\in\widehat{\Gamma}$ by the proof of Theorem~\ref{thm:construction_pr} since $p\geq 2w-1$.
So, it suffices to consider the case when $a,b\in\widehat{Q}$.

Notice that $d^*(S_{(1,g)})=\{\pm j(1,g)\in\mathbb{Z}_{w}\times\mathbb{Z}_{p^r}:\,j=1,2,\ldots,w-1\}$.
Assume $j(1,g)=\pm i(1,h)\in\mathbb{Z}_w\times\mathbb{Z}_{p^r}$ for some $g\neq h\in\mathcal{S}_r(Q)$ and $1\leq i,j\leq w-1$.
There are two cases $i=j$ and $j=-i$ according to the first component.
The former case yields a contradiction that $g=h$.
So, it suffices to consider the case that $j(1,g)=-i(1,h)$ in $\mathbb{Z}_w\times\mathbb{Z}_{p^r}$.
The two components indicate $i+j=w$ and $jg+ih= 0\ (\bmod\ p^r)$, which imply that $ih=(i-w)g\ (\bmod\, p^r)$.
That is, by considering the $p$-ary representations, both $ih$ and $(i-w)g$ are in the same layer, say $\mathsf{L}_t$ for some $t$.
Since $i\leq w-1<p-1$, we have $i=i_0$ and $i-w=(i-w)_0$, namely, both $i$ and $i-w$ are in $\mathsf{L}_0$.
It follows from Proposition~\ref{prop:p-ary} that $g,h\in\mathsf{L}_t$.
More precisely, 
\begin{align*}
h=i^{-1}\cdot ih\in\mathsf{L}_t \quad \text{and} \quad g=(i-w)^{-1}\cdot (i-w)g\in\mathsf{L}_t.
\end{align*}
Therefore, $g_t,h_t\in Q$ by assumption.
Then, by \eqref{eq:Legendre} and Proposition~\ref{prop:p-ary}, we have
\begin{align*}
& (ih)_t = ((i-w)g)_t \ (\bmod\, p) \\
\Rightarrow \ & i\cdot h_t = (i-w)\cdot g_t \ (\bmod\, p) \\
\Rightarrow \ & \left(\frac{i}{p}\right) \left(\frac{h_t}{p}\right) = \left(\frac{i-w}{p}\right) \left(\frac{g_t}{p}\right) \\
\Rightarrow \ & \left(\frac{i}{p}\right)  = \left(\frac{i-w}{p}\right),
\end{align*}
where the last implication is due to $g_t,h_t\in Q$. 
This contradicts the condition given in~\eqref{eq:quadratic-non-residue}, and the proof is completed.
\end{proof}

\begin{example}\label{ex:4x47r4} \rm
Let $p=47,w=4$. 
One can check that $\Gamma=\{1,4,11,19,20,21\}$ forms a set of generators of a code in $\CACe(47,4)$.
Notice that the set of quadratic residues modulo $47$ is $Q=\{1,2,3,4,6,7,8,9,12,14,16,17,18,21,24,25,27,28,32,34,36,37,42\}$.
Obviously, $\left(\frac{1}{47}\right)\left(\frac{-3}{47}\right)=\left(\frac{2}{47}\right)\left(\frac{-2}{47}\right)=-1$.
By Theorem~\ref{thm:construction_wpr}, we have an equi-difference CAC of length $4\cdot 47^r$ and weight $4$ with $(47^r-1)/6 + (47^r-1)/2 + 1$ codewords, for each integer $r\geq 1$.
When $r=1$, we have $\widehat{\Gamma}=\{(0,g)\in\mathbb{Z}_4\times\mathbb{Z}_{47}:\,g\in\Gamma\}$ and $\widehat{Q}=\{(1,g)\in\mathbb{Z}_4\times\mathbb{Z}_{47}:\,g\in Q\}$.
So, the obtained code in $\CACe(188,4)$ has generators in $\{\theta^{-1}(a):a\in\widehat{\Gamma}\uplus\widehat{Q}\uplus\{(1,0)\}\}$, where the bijection $\theta:\mathbb{Z}_{188}\to\mathbb{Z}_4\times\mathbb{Z}_{47}$ is given in \eqref{eq:CRT-correspondence}.
The corresponding generators are listed as follows.
\begin{align*}
\theta^{-1}(\widehat{\Gamma}) &= \{4,20,48,68,152,160\}, \\
\theta^{-1}(\widehat{Q}) &= \{1,9,17,21,25,37,49,53,61,65,81,89,97,\\ 
			& \ \ \ \ \ 101,121,145,149,153,157,165,169,173,177\}, \text{ and}\\
\theta^{-1}(\{(1,0)\}) &= \{141\}.
\end{align*}
\end{example}

\begin{theorem}\label{thm:main_wpr}
Let $p$ be a prime such that $p-1$ is divisible by $2w-2$.
If there is a code in $\CACe(p,w)$ with $(p-1)/(2w-2)$ codewords and the condition \eqref{eq:quadratic-non-residue} in Theorem~\ref{thm:construction_wpr} holds, then for any integer $r\geq 1$, 
\begin{align*}
K\left(wp^r,w\right)=\frac{p^r-1}{2w-2}+\frac{p^r-1}{2}+1.
\end{align*}
\end{theorem}
\begin{proof}
The assumption that $p-1$ is divisible by $2w-2$ guarantees $p\geq 2w-1$.
By Theorem~\ref{thm:construction_wpr}, there exists a code in $\CACe(wp^r,w)$ with $\frac{p^r-1}{2w-2}+\frac{p^r-1}{2}+1$ codewords.
It suffices to show $K\left(wp^r,w\right)\leq \frac{p^r-1}{2w-2}+\frac{p^r-1}{2}+1$.

Let $\C$ be any code in $\CAC(wp^r,w)$.
Let $\mathcal{E}\subseteq\C$ be the collection of all exceptional codewords. 
Following the notation in the proof of Theorem~\ref{thm:main_w-1dpr}, denote by $H_S=\mathsf{H}(d(S))$ and $H_S^*=H_S\setminus\{0\}$ for $S\in\mathcal{E}$.

It follows from Corollary~\ref{cor:excpetional-bound} that $|H_S|\leq 2w-2< p$ for any $S\in\mathcal{E}$, which implies that $\gcd(|H_S|,p)=1$.
Moreover, $|H_S|$ divides $wp^r$ since $H_S$ is a subgroup of $\mathbb{Z}_{wp^r}$.
Hence $|H_S|$ must divide $w$.
On the other hand, as $H_S$ is a subgroup of $\mathbb{Z}_{wp^r}$, we have $H_S=-H_S$, which implies that $|-S+H_S|=|-(S+H_S)|=|S+H_S|$.
By plugging $A=S,B=-S$ into \eqref{eq:Kneser1},
\begin{align*}
|d(S)|=|S+(-S)| &\geq |S+H_S|+|-S+H_S|-|H_S| \\
&= 2|S+H_S|-|H_S| \geq 2|S|-|H_S^*|-1, 
\end{align*}
which yields 
\begin{equation}\label{eq:main_wpr_case2-1}
|d^*(S)|\geq 2|S|-2-|H_S^*| = 2w-2-|H_S^*|.
\end{equation}

We now claim that $\sum_{S\in\mathcal{E}}|H^{*}_S|\leq w-1$.
Since $0\in d(S)$ for $S\in\mathcal{E}$, it follows from Proposition~\ref{prop:stabilizer}(ii) that $H_S\subseteq d(S)$.
Then, $H_S^*\cap H_{S'}^*=\emptyset$ for any two distinct $S,S'\in\mathcal{E}$ because of $d^*(S)\cap d^*(S')=\emptyset$.
Moreover, since $H_S$ is a subgroup of $\mathbb{Z}_{wp^r}$ and $|H_S|$ divides $w$, by Proposition~\ref{prop:unique_subgroup}, $H_S$ is a subgroup of $G=\{ip^r:\,i=0,1,\ldots,w-1\}$.
This concludes that
\begin{equation}\label{eq:main_wpr_case2-2}
\sum_{S\in\mathcal{E}}|H_S^*| = \left|\biguplus_{S\in\mathcal{E}}H_S^*\right| \leq |G\setminus\{0\}| = w-1.
\end{equation}
Combining \eqref{eq:main_wpr_case2-1}--\eqref{eq:main_wpr_case2-2} yields
\begin{equation}\label{eq:main_wpr_case2-3}
\sum_{S\in\mathcal{E}}|d^*(S)| \geq (2w-2)|\mathcal{E}| - (w-1).
\end{equation}

By the disjoint-difference-set property and \eqref{eq:main_wpr_case2-3}, we have
\begin{align*}
wp^r-1 = |\mathbb{Z}^*_{wp^r}| &\geq \sum_{S\in\C\setminus\mathcal{E}}|d^*(S)| + \sum_{S\in\mathcal{E}}|d^*(S)| \\
&\geq (2w-2)(|\C|-|\mathcal{E}|) + (2w-2)|\mathcal{E}| - (w-1) \\
&= (2w-2)|\C| - (w-1),
\end{align*}
and thus
\begin{align*}
|\C| \leq \left\lfloor\frac{wp^r+w-2}{2w-2}\right\rfloor = \left\lfloor \frac{p^r-1}{2} + \frac{p^r-1}{2w-2} + \frac{2w-2}{2w-2} \right\rfloor = \frac{p^r-1}{2} + \frac{p^r-1}{2w-2} + 1.
\end{align*}
\end{proof}

Analogous to Corollary~\ref{coro:w-1prw}, the primes that satisfy the condition in~\eqref{eq:quadratic-non-residue} for some small $w$ are listed in the following table.

\begin{center}
\begin{tabular}{c|l|l}
$w$ & simplified equations of \eqref{eq:quadratic-non-residue} & $p$ satisfies the condition in~\eqref{eq:quadratic-non-residue} \\ \hline
$3$ & $\big(\frac{-2}{p}\big)=-1$ & $p\equiv -1,-3\ (\bmod\ 8)$\\
$4$ & $\big(\frac{-1}{p}\big)=-1$ and  $\big(\frac{3}{p}\big)=1$ & $p\equiv -1\ (\bmod\ 12)$ \\
$5$ & $\big(\frac{-1}{p}\big)=-1$ and $\big(\frac{6}{p}\big)=1$ & $p\equiv -1,-5\ (\bmod\ 24)$ \\
$6$ & $\big(\frac{-1}{p}\big)=-1$ and $\big(\frac{2}{p}\big)=\big(\frac{5}{p}\big)=1$ & $p\equiv -1,-9\ (\bmod\ 40)$ \\
$7$ & $\big(\frac{2}{p}\big)=1$ and $\big(\frac{-3}{p}\big)=\big(\frac{-5}{p}\big)=-1$ & $p\equiv -1,,-7,17,-49\ (\bmod\ 120)$  \\
$8$ & $\big(\frac{-1}{p}\big)=-1$ and $\big(\frac{3}{p}\big)=\big(\frac{5}{p}\big)=\big(\frac{7}{p}\big)=1$ & $p\equiv -1,59,-109,-121,131,-169\ (\bmod\ 420)$  \\
$9$ & $\big(\frac{-2}{p}\big)=\big(\frac{-5}{p}\big)=-1$ and $\big(\frac{7}{p}\big)=1$ & $p\equiv -1,-3,-9,,-27,31,37,53,-81,-83,$ \\
 & & \qquad $93,111,-121\ (\bmod\ 280)$  \\
$10$ & $\big(\frac{-1}{p}\big)=-1$ and $\big(\frac{6}{p}\big)=\big(\frac{3}{p}\big)\big(\frac{7}{p}\big)=1$ & $p\equiv -1,-5,-25,43,47,67\ (\bmod\ 168)$  \\
\end{tabular}
\end{center}

\bigskip

\subsection{Optimal CACs of length $(2w-1)p^r$}


\begin{theorem}\label{thm:construction_2w-1pr}
Let $p$ be a prime such that $p>2w-1$.
If there is a code in $\CACe(p,w)$ with $m$ codewords, then for any integer $r\geq 1$, there exists a code $\C\in\CACe((2w-1)p^r,w)$ with $|\C|=p^r+m(p^r-1)/(p-1)$ codewords.
\end{theorem}
\begin{proof}
Let $\Gamma$ be a set of $m$ generators of a given code in $\CACe(p,w)$.
Define 
\begin{align*}
\widehat{\Gamma}\triangleq\{(0,g)\in\mathbb{Z}_{2w-1}\times\mathbb{Z}_{p^r}:g\in\mathcal{S}_r(\Gamma)\}
\end{align*}
and
\begin{align*}
\Lambda\triangleq\{(1,g)\in\mathbb{Z}_{2w-1}\times\mathbb{Z}_{p^r}:\,0\leq g\leq p^r-1\}.
\end{align*}
Obviously, $\widehat{\Gamma}$ and $\Lambda$ are disjoint.
We shall prove $\widehat{\Gamma}\uplus\Lambda$ is the set of generators of the desired code $\C$.

Since $p$ is a prime with $p>2w-1$, we have $\gcd(2w-1,p)=1$ and thus $\Z_{(2w-1)p^r}\cong\Z_{2w-1}\times\Z_{p^r}$.
For $a\in\widehat{\Gamma}\uplus\Lambda$, define $S_g=\{jg:\,j=0,1,\ldots,w-1\}$.
We will show $d^*(S_a)\cap d^*(S_b)=\emptyset$ for $a\neq b\in\widehat{\Gamma}\uplus\Lambda$.
The assertion is obviously true in the case when $a\in\widehat{\Gamma},b\in\Lambda$.
It also holds in the case when both $a,b\in\widehat{\Gamma}$ by the proof of Theorem~\ref{thm:construction_pr} since $p>2w-1$.
So, it suffices to consider the case when $a,b\in\Lambda$.

Notice that $d^*(S_{(1,g)})=\{\pm j(1,g)\in\mathbb{Z}_{2w-1}\times\mathbb{Z}_{p^r}:\,j=1,2,\ldots,w-1\}$.
Assume $j(1,g)=\pm i(1,h)$ for some $g\neq h$ and $1\leq i,j\leq w-1$.
There are two cases $i=j$ and $i=-j$ (i.e., $i=2w-1-j$) according to the first component.
The former case yields a contradiction that $g=h$, while the latter one also implies a contradiction that $i\geq (2w-1)-(w-1)=w$ due to $j\leq w-1$.

Finally, by \eqref{eq:S_r(A)-size}, we have
\begin{align*}
|\C|=|\Lambda| + |\mathcal{S}_r(\Gamma)| = p^r + m(1+p+\cdots+p^{r-1})  = p^r+\frac{m(p^r-1)}{p-1}.
\end{align*}
\end{proof}

\begin{remark}\rm
The proof of $d^*(S_a)\cap d^*(S_b)=\emptyset$ in Theorem~\ref{thm:construction_2w-1pr} for the case that $a,b$ are distinct elements in $\Lambda$ can be found in~\cite{SWSC09}.
\end{remark}

\begin{example}\label{ex:7x37r4} \rm
Let $p=37,w=4$. 
Following Example~\ref{ex:37r4}, $\Gamma=\{1,6,8,10,11,14\}$ forms a code in $\CACe(37,4)$ of size $6$.
By Theorem~\ref{thm:construction_2w-1pr}, we have an equi-difference CAC of length $7\cdot {37}^r$ and weight $4$ with $37^r+(37^r-1)/6$ codewords, for each integer $r\geq 1$.
When $r=1$, we have $\widehat{\Gamma}=\{(0,g)\in\mathbb{Z}_7\times\mathbb{Z}_{37}:\,g\in\Gamma\}$ and $\Lambda=\{(1,g)\in\mathbb{Z}_7\times\mathbb{Z}_{37}:\,0\leq g\leq 36\}$.
So, the obtained code in $\CACe(259,4)$ has generators in $\theta^{-1}(a):a\in\widehat{\Gamma}\uplus\Lambda$, where the bijection $\theta:\mathbb{Z}_{259}\to\mathbb{Z}_7\times\mathbb{Z}_{37}$ is given in \eqref{eq:CRT-correspondence}.
The generators produced from $\widehat{\Gamma}$ are $14,84,112,119,154,196$, and from $\Lambda$ are 
\begin{align*}
&1,8,15,22,29,36,43,50,57,64,71,78,	85,92,99,106,113,120,127,134,141,\\
&148,155,162,169,176,183,190,197,204,211,218,225,232,239,246,253.
\end{align*}
\end{example}


\begin{theorem}\label{thm:main_2w-1pr}
Let $p$ be a prime such that $p-1$ is divisible by $2w-2$. 
If there is a code in $\CACe(p,w)$ with $(p-1)/(2w-2)$ codewords, then for any integer $r\geq 1$,
\begin{align*}
K\left((2w-1)p^r,w\right)=p^r+\frac{p^r-1}{2w-2}.
\end{align*}
\end{theorem}
\begin{proof}
The assumption that $p-1$ is divisible by $2w-2$ guarantees that $p\geq 2w-1$.
The case when $p=2w-1$ can be reduced to Theorem~\ref{thm:pr-recall-optimal}, i.e., $K(p^{r+1},(p+1)/2)=(p^{r+1}-1)/(p-1)$.
So, we may assume $p>2w-1$ in the followings.

As $p>2w-1$, by Theorem~\ref{thm:construction_2w-1pr}, there exists a code in $\CACe((2w-1)p^r,w)$ with $p^r+(p^r-1)/(2w-2)$ codewords.
Therefore, it suffices to show $K\left((2w-1)p^r,w\right)\leq p^r+(p^r-1)/(2w-2)$.

Assume $\C\in\CAC((2w-1)p^r,w)$.
We shall claim that every codeword in $\C$ is non-exceptional.
Suppose to the contrary that $S\in\C$ is exceptional.
By Corollary~\ref{cor:excpetional-bound}, $|\HH(d(S))|\leq 2w-2<p$, namely $\gcd(|\HH(d(S))|,p)=1$.
As $|\HH(d(S))|$ divides $(2w-1)p^r$ due to $\HH(d(S))$ a subgroup of $\mathbb{Z}_{(2w-1)p^r}$, it follows that $|\HH(d(S))|$ divides $2w-1$.
By Lemma~\ref{lem:nonexceptional}, $S$ is non-exceptional, a contradiction occurs.
%
By the disjoint-difference-set property, $|\Z^*_{(2w-1)p^r}|\geq \sum_{S\in\C}|d^*(S)| \geq (2w-2)|\C|$, yielding
\begin{align*}
|\C| \leq \frac{(2w-1)p^r-1}{2w-2} = p^r+\frac{p^r-1}{2w-2},
\end{align*}
as desired.
\end{proof}


The case when $w\leq p<2w-1$ is missing in Theorem~\ref{thm:main_2w-1pr}.
We will fill the gap in the following theorem, which is an improvement of~\cite[Theorem 14]{SWC10}.

\begin{theorem}\label{thm:main_2w-1p}
Let $p$ be a prime such that $w\leq p<2w-1$. 
One has
\begin{align*}
K((2w-1)p,w)=p+1.
\end{align*}
\end{theorem}
\begin{proof}
Let $\Lambda=\{(1,g)\in\Z_{2w-1}\times\Z_p:\,0\leq g\leq p-1\}$.
It follows from the proof of Theorem~\ref{thm:construction_2w-1pr} that $\Lambda\cup\{(0,1)\}$ forms a set of generators of a code in $\CAC((2w-1)p,w)$ with $p+1$ codewords.

It suffices to show that for any code $\C\in\CAC((2w-1)p,w)$, one has $|\C|\leq p+1$.
Let $S\in\C$ be an exceptional codeword (if any), and denote by $H_S=\HH(d(S))$.
Note that $2w-1$ and $p$ are relatively prime due to $w\leq p<2w-1$.
By the fact that $H$ is a subgroup of $\Z_{(2w-1)p}$, we have $|H_S|$ divides either $2w-1$ or $p$.
However, it must be the case that $H_S$ divides $p$ because the other one contradicts to the assertion in Lemma~\ref{lem:nonexceptional}.
As shown in the Case 2 in the proof of Theorem~\ref{thm:main_w-1dpr}, the exceptional codeword is unique and satisfies $|d^*(S)|\geq p-1$.
Hence it concludes that there is at most one exceptional codeword in $\C$, and the unique codeword $S$ is of the property that $|d^*(S)|\geq p-1$.
By the same argument in the proof of Theorem~\ref{thm:main_w-1dpr}, we have $|\C|\leq p+1$, as desired.
\end{proof}

\begin{remark}\rm
\cite[Theorem 14]{SWC10} claimed that the equality $K((2w-1)p,w)=p+1$ holds when $w\leq p\leq w+\frac{\varphi(2w-1)}{2}$, where $\varphi(n)$ is the Euler's totient function, i.e., counts the number of integers up to $n$ that are relatively prime to $n$. 
\end{remark}

\section{Mixed-Weight CACs}\label{sec:MW-CAC}


By the help of the construction given in Theorem~\ref{thm:construction_direct}, in this subsection we first propose a general construction of a mixed-weight CAC of length $(w-1)p^r$ with weight-set $\{w-1,w,w^*\}$, where $p$ is an odd prime and $r,w,w^*$ are any positive integers with $p\geq w$.
Based on this construction, we derive the exact value of $K\left((w-1)p^r,w-1;w,n\right)$ for some $n$.


%

Recall that when $d=1$, the two conditions \eqref{eq:SDR1} and \eqref{eq:SDR2} are respectively reduced to \eqref{eq:QR1} and \eqref{eq:QR2}, i.e.,
\begin{align*}
\left(\frac{-1}{p}\right)=-1
\end{align*}
and
\begin{align*}
\left(\frac{i}{p}\right)\left(\frac{i-w+1}{p}\right)=-1,\ \forall i=1,2,\ldots,w-2.
\end{align*}

\begin{theorem}\label{thm:mixed-w-1pr}
Let $r,w$ be positive integers and $p$ be an odd prime such that $p\geq w$.
Suppose $p$ and $w$ satisfy the two conditions given in \eqref{eq:QR1} and \eqref{eq:QR2}, and there exists a code $\mathcal{A}\in\text{CAC}(p^r,w^*)$ that contains $n$ equi-difference codewords, where $w^*$ is an arbitrary positive integer.
Then, there exists a code $\mathcal{C}\in\text{CAC}\left((w-1)p^r,\{w-1,w,w^*\}\right)$ with $|\mathcal{C}|=\frac{p^r-1}{2}+n+1$ codewords.
In particular, if the $n$ equi-difference codewords in $\mathcal{A}$ are all not exceptional, then $\mathcal{C}$ contains $n$ codewords with weight $w^*$, $\frac{p^r-1}{2}-n(w^*-1)$ codewords with weight $w$, and $n(w^*-1)+1$ codewords with weight $w-1$.
\end{theorem}
\begin{proof}
Let $A_1,\ldots,A_n$ be the $n$ equi-difference codewords in $\mathcal{A}$. 
We only consider the case that each $A_i$ is not exceptional, since the other cases can be dealt with in the same way.
Assume the generator of $A_i$ is $a_i$ for $i=1,\ldots,n$.
By definition, $A_i=\{0,a_i,\ldots,(w^*-1)a_i\}$ and $d^*(A_i)=\{\pm a_i,\ldots,\pm(w^*-1)a_i\}$ for all $i$, and $d^*(A_i)\cap d^*(A_j)=\emptyset$ for any two distinct $i,j$.

Recall that $H^2(p)=Q(p)$.
Let $Q=Q(p)$ for the sake of notational convenience.
For convenience, let $\mathcal{S}_r(Q)=Q_0\uplus Q_1\uplus\cdots\uplus Q_{r-1}$, where $Q_t=\{c\in\mathsf{L}_t:\,c_t\in Q\}$.
Consider the code $\mathcal{C}'\in\text{CAC}((w-1)p^r,w)$ obtained in Theorem~\ref{thm:construction_direct} consists of equi-difference codewords
\begin{equation}\label{eq:mix_w-1pr-Sg}
S_g=\{j(1,g)\in\mathbb{Z}_{w-1}\times\mathbb{Z}_{p^r}:\,j=0,1,2,\ldots,w-1\},\ \forall g\in\mathcal{S}_r(Q).
\end{equation}
Note that the difference set of $S_g$ is in the form
\begin{equation}\label{eq:mix_w-1pr-d*}
d^*(S_g)=\{\pm j(1,g)\in\mathbb{Z}_{w-1}\times\mathbb{Z}_{p^r}:\,j=1,2,\ldots,w-2\}\cup\{0,\pm(w-1)g\}.
\end{equation}

We will obtain three classes of codewords, say $\mathcal{C}_{w^*}, \mathcal{C}_{w}$ and $\mathcal{C}_{w-1}$, consist of codewords with weights $w^*,w$ and $w-1$, respectively.
The main idea is, for each codeword in $\mathcal{A}$, to associate some $w^*-1$ codewords in $\mathcal{C}'$ and reconstruct them to obtain one $w^*$-weight codeword and $w^*-1$ $(w-1)$-weight codewords.

Firstly, let $\mathcal{C}_{w^*}=\{T_{a_1},T_{a_2},\ldots,T_{a_n}\}$, where
\begin{align*}
T_{a_i} = \{(0,0),(0,a_i),(0,2a_i),\ldots,(0,(w^*-1)a_i)\},
\end{align*}
for $i=1,\ldots,n$.
Observe that 
\begin{equation}\label{eq:difference-w*}
d^*(T_{a_i}) = \{(0,\pm a_i), (0,\pm 2a_i), (0,\pm (w^*-1)a_i)\}.
\end{equation}
For $i\neq j$, since $d^*(A_i)\cap d^*(A_j)=\emptyset$, it is easy to see that 
\begin{equation}\label{eq:difference-disjoint-w*}
d^*(T_{a_i})\cap d^*(T_{a_j})=\emptyset.
\end{equation}

Secondly, fix any $1\leq i\leq n$.
For each $k\in\{1,2,\ldots,w^*-1\}$, since $\left(\frac{-1}{p}\right)=-1$, it is not hard to see that exactly one of $ka_i(w-1)^{-1}$ or $-ka_i(w-1)^{-1}$ is in $Q_{t}$, for some $0\leq t\leq r-1$.
Here, $(w-1)^{-1}$ indicates the multiplicative inverse of $w-1$ in the multiplicative group $\mathbb{Z}^{\times}_{p^r}$, and its existence is guaranteed by $w-1<p$.
Let $g_{i_k}\in\{ka_i(w-1)^{-1},-ka_i(w-1)^{-1}\}$ be the quadratic residue.
Observe that $S_{g_{i_k}}$ is a codeword in $\mathcal{C}'$ with difference set
\begin{align*}
d^*(S_{g_{i_k}}) = \{\pm j(1,g_{i_k}):\, j=1,2,\ldots,w-2\} \cup \{(0,\pm ka_i)\}
\end{align*}
due to \eqref{eq:mix_w-1pr-d*} and $\pm (w-1)g_{i_k}=\pm ka_i$.
Now, for $i=1,\ldots,n$ and $k=1,2,\ldots,w^*-1$, let 
\begin{align*}
S'_{g_{i_k}} = S_{g_{i_k}} \setminus \{(w-1)(1,g_{i_k})\},
\end{align*}
whose difference set would be
\begin{equation}\label{eq:difference-w}
d^*(S'_{g_{i_k}}) = d^*(S_{g_{i_k}}) \setminus \{(0,\pm ka_i)\}.
\end{equation}
Let 
\begin{align*}
G = \{g_{i_k}:\,i=1,\ldots,n\text{ and }k=1,2,\ldots,w^*-1\}
\end{align*}
be the collection of the generators considered here.
It follows from \eqref{eq:difference-w*} and \eqref{eq:difference-w} that $d^*(S'_g)\cap d^*(T)=\emptyset$ for $g\in G$ and $T\in\mathcal{C}_{w^*}$.
Moreover, define
\begin{equation}\label{eq:difference-S'0}
S'_0=\{(j,0)\in\mathbb{Z}_{w-1}\times\mathbb{Z}_{p^r}:\,j=0,1,\ldots,w-2\},
\end{equation}
and let 
\begin{align*}
\mathcal{C}_{w-1} = \left\{S'_0\}\cup\{S'_{g}:\,g\in G\right\}.
\end{align*}
Observe that the differences in $d^*(S'_0)$ are all of the form $(\pm j,0)$, for $j=1,2,\ldots, w-2$, and thus $d^*(S'_0)\cap\bigcup_{S\in\C'}d^*(S)=\emptyset$, due to \eqref{eq:mix_w-1pr-d*}.
Hence the difference sets of codewords in $\mathcal{C}_{w^*}\cup\mathcal{C}_{w-1}$ are mutually disjoint.

Finally, let 
\begin{align*}
\mathcal{C}_{w} = \mathcal{C}'\setminus\{S_g:\,g\in G\}.
\end{align*}
By the assumption that $\C'\in\text{CAC}\left((w-1)p^r,w\right)$, the set $\C=\C_{w^*}\cup\C_{w-1}\cup\C_w$ forms a code in $\text{CAC}\left((w-1)p^r,\{w^*,w,w-1\}\right)$, as desired.
\end{proof}

One can apply the construction given in Theorem~\ref{thm:construction_direct} iteratively to construct a mixed-weight CAC with various weights.
In other words, if the based code $\mathcal{A}$ is a mixed-weight CAC with weight set $\{w^*_1,\ldots,w^*_t\}$, then the resulting mixed-weight CAC is with weight set $\{w-1,w\}\cup\{w^*_1,\ldots,w^*_t\}$.
Note that $w^*_i$, $1\leq i\leq t$, may be identical to $w$ or $w-1$.

The following example illustrates our idea in the proof of Theorem~\ref{thm:mixed-w-1pr}.

\begin{example}\label{ex:mixed-1} \rm
Let $p=23,r=1,w=4,w^*=7$ and $n=1$.
One has $L=p^r=23$.
The set of quadratic residues in $\mathbb{Z}_{23}$ is $Q(23)=\{1,2,3,4,6,8,9,12,13,16,18\}$.
One can check that
\begin{align*}
\left(\frac{-1}{23}\right) = \left(\frac{1}{23}\right)\left(\frac{-2}{23}\right)=-1,
\end{align*}
that satisfy the conditions in \eqref{eq:QR1} and \eqref{eq:QR2}.

By the CRT correspondence~\eqref{eq:CRT-correspondence}, the elements $(1,g)$, $g\in Q(23)$, in $\mathbb{Z}_3\times\mathbb{Z}_{23}$ are $1$, $25$, $49$, $4$, $52$, $31$, $55$, $58$, $13$, $16$ and $64$ in $\mathbb{Z}_{69}$, respectively.
Then, the code in $\CACe(69,4)$ obtained by the construction in Theorem~\ref{thm:construction_direct} (or, \cite[Theorem 3]{SW10} since $r=1$) contains the following $11$ codewords:

\begin{center}
\begin{tabular}{llll}
$S_1=\{0,1,2,3\},$ & $S_2=\{0,25,50,6\},$ & $S_3=\{0,49,29,9\},$ & $S_4=\{0,4,8,12\},$ \\
$S_6=\{0,52,35,18\}$,  & $S_8=\{0,31,62,24\}$,  & $S_9=\{0,55,41,27\}$,  & $S_{12}=\{0,58,47,36\}$,  \\
$S_{13}=\{0,13,26,39\}$,  & $S_{16}=\{0,16,32,48\}$,  & $S_{18}=\{0,64,59,54\}$  &   \\
\end{tabular}
\end{center}

Consider $\mathcal{A}=\{A_1=\{0,1,2,3,4,5,6\}\}$ a CAC of length $23$ with weight $w^*=7$ containing only one element.
Define $\C_7=\{T_1\}$ by
\begin{align*}
T_1 &= \{(0,k)\in\mathbb{Z}_3\times\mathbb{Z}_{23}:\,k=0,1,\ldots,6\} = \{0,24,48,3,27,51,6\} \subseteq\mathbb{Z}_{69},
\end{align*}
where the last identity is due to the CRT correspondence.

As $w^{-1}=3^{-1}=8$ in the multiplicative group $\mathbb{Z}^{\times}_{23}$, the elements $kw^{-1}$ and $-kw^{-1}$ for $k=1,\ldots,6$ are listed as follows, where the bold face refers to an element in $Q(23)$.
\begin{center}
\begin{tabular}{c|cccccc}
$k$ & 1 & 2 & 3 & 4 & 5 & 6 \\ \hline
$kw^{-1}$ & \textbf{8} & \textbf{16} & \textbf{1} & \textbf{9} & 17 & \textbf{2} \\ \hline
$-kw^{-1}$ & 15 & 7 & 22 & 14 & \textbf{6} & 21 \\ 
\end{tabular}
\end{center}
Therefore, $G=\{1,2,6,8,9,16\}$, and thus $\mathcal{C}_3$ contains
\begin{center}
\begin{tabular}{lll}
$S'_{1}=\{0,1,2\}$, & $S'_{2}=\{0,25,50\}$, & $S'_{6}=\{0,52,35\}$, \\
$S'_8=\{0,31,62\}$, & $S'_{9}=\{0,55,41\}$, & $S'_{16}=\{0,16,32\}$, 
\end{tabular}
\end{center}
and the extra one $S'_0=\{0,46,23\}$.
Finally, the codewords with weight $w=4$ are $S_3, S_4, S_{12}, S_{13}$ and $S_{18}$.
\end{example}


\begin{theorem}\label{thm:mixed-w-1pr-optimal}
Let $r,w$ be positive integers and $p$ be an odd prime such that $p\geq 2w-1$.
Suppose $p$ and $w$ satisfy the two conditions given in \eqref{eq:QR1} and \eqref{eq:QR2}.
For $w^*=w-1$ or $w$, if there exists a code $\mathcal{A}\in\text{CAC}(p^r,w^*)$ that contains $n$ ($n\leq\lfloor\frac{p^r-1}{2(w^*-1)}\rfloor$) equi-difference codewords, then 
\begin{align*}
K\left((w-1)p^r,w-1;w,n'\right)=n+\frac{p^r+1}{2},
\end{align*}
where $n'=\frac{p^r-1}{2}-n(w-2)$.
\end{theorem}
\begin{proof}
Pick any $S\in\mathcal{A}$.
If $S$ is exceptional, by Corollary~\ref{cor:excpetional-bound}, $|\mathsf{H}(d(S))|\leq 2w^*-2$, which is less than $p$ due to $w-1\leq w^*\leq w$ and $p\geq 2w-1$.
This indicates that $\gcd(|\mathsf{H}(d(S))|,p)=1$.
Since $\mathsf{H}(d(S))$ is a subgroup of $\mathbb{Z}_{p^r}$, it follows that $|\mathsf{H}(d(S))|=1$, which is a contradiction to the assertion in Corollary~\ref{cor:excpetional-bound} that $|\HH(d(S))|\geq 2$.
Therefore, the $n$ equi-difference codewords in $\mathcal{A}$ are all not exceptional.

Let $\C=\C_w\cup\C_{w-1}\in\CAC((w-1)p^r,\{w-1,w\})$ be the resulting mixed-weight CAC by plugging $w^*=w$ or $w-1$ into the construction of Theorem~\ref{thm:mixed-w-1pr}, where $\C_w$ (resp., $\C_{w-1}$) refers to the set of codewords with weight $w$ (resp., $w-1$).
One can check that $|\C_w|=\frac{p^r-1}{2}-n(w-2)$ and $|\C_{w-1}|=n(w-1)+1$.
So, it suffices to show that $K\left(wp^r,w;w+1,n'\right)\leq n+\frac{p^r+1}{2}$.

Let $\C'=\C'_w\uplus\C'_{w-1}\in\CAC\left((w-1)p^r,\{w-1,w\}\right)$ be any mixed-weight CAC, where $C'_w$ (resp., $C'_{w-1}$) consists of all $w$-weight (resp., $(w-1)$-weight) codewords, and $|C'_w|=n'=\frac{p^r-1}{2}-n(w-2)$.
Let $\mathcal{E}\subseteq\C'$ be the collection of all exceptional codewords.
The \textit{Case 1} in the proof of Theorem~\ref{thm:main_w-1dpr} shows that any codeword with weight $w$ is non-exceptional.
That is, $\mathcal{E}\subseteq\C'_{w-1}$.
For $S\in\mathcal{E}$, by Corollary~\ref{cor:excpetional-bound}, one has $\HH(d(S))|\leq 2w-4<p$.
Since $\HH(d(S))$ is a subgroup of $\mathbb{Z}_{(w-1)p^r}$, it follows that $\gcd(|\HH(d(S))|,p)=1$, and thus $|\mathsf{H}(d(S))|$ divides $w-1$.
By the same argument as in the derivation of \eqref{eq:main_wpr_case2-3} with placing $w$ by $w-1$, we have $|d^*(S)|\geq 2w-4-|H^*_S|$ and $\sum_{S\in\mathcal{E}}|H^*_S|\leq w-2$, where $H^*_S=\HH(d(S))\setminus\{0\}$.
Therefore,
\begin{align*}
\sum_{S\in\mathcal{E}}|d^*(S)| \geq (2w-4)|\mathcal{E}|-(w-2).
\end{align*}
By the disjoint-difference-set property, 
\begin{align*}
(w-1)p^r = |\mathbb{Z}^*_{(w-1)p^r}| \geq & \sum_{S\in\C'_w}|d^*(S)| + \sum_{S\in\C'_{w-1}\setminus\mathcal{E}}|d^*(S)| + \sum_{S\in\mathcal{E}}|d^*(S)| \\
\geq & \ (2w-2)\left(\frac{p^r-1}{2}-n(w-2)\right) \\
& + (2w-4)\left(|\C'|-\left(\frac{p^r-1}{2}-n(w-1)\right)-|\mathcal{E}|\right) \\ 
& + (2w-4)|\mathcal{E}|-(w-2) \\
= & (2w-4)|\C'| + p^r-1 -(w-2)(2n+1),
\end{align*}
yielding that
\begin{align*}
|\C'| \leq \left\lfloor\frac{(w-2)p^r+(w-2)(2n+1)}{2w-4}\right\rfloor = \frac{p^r-1}{2}+n+1,
\end{align*}
as desired.
\end{proof}

%



Let us turn back to the constructions in Theorems~\ref{thm:construction_wpr} and \ref{thm:construction_2w-1pr}.
Let the based CAC be a code in $\CACe(p,w^*)$, for some $w^*\neq w$, and $\Gamma$ be the set of $m$ generators.
By defining the corresponding set of generators as $\widehat{\Gamma}=\{(0,g)\in\mathbb{Z}_{w}\times\mathbb{Z}_{p^r}:\,g\in\mathcal{S}_r(\Gamma)\}$, we get the following two consequences.


\begin{corollary}\label{cor:mixed-wpr}
Let $p$ be a prime such that $p\geq 2w-1$.
Assume $w^*$ is an arbitrary positive integer.
If there is a code in $\CACe(p,w^*)$ with $m$ codewords and the condition in \eqref{eq:quadratic-non-residue} holds, then for any integer $r\geq 1$, there exists a code $\mathcal{C}\in\CAC(wp^r,\{w,w^*\})$ with $(p^r+1)/2$ codewords of weight $w$ and $m(p^r-1)/(p-1)$ codewords of weight $w^*$.
\end{corollary}

\begin{corollary}\label{cor:mixed-2w-1pr}
Let $p$ be a prime such that $p>2w-1$.
Assume $w^*$ is an arbitrary positive integer.
If there is a code in $\CACe(p,w^*)$ with $m$ codewords, then for any integer $r\geq 1$, there exists a code $\mathcal{C}\in\CAC((2w-1)p^r,\{w,w^*\})$ with $p^r$ codewords of weight $w$ and $m(p^r-1)/(p-1)$ codewords of weight $w^*$.
\end{corollary}

Finally, we have the following two classes of optimal mixed-weight CACs of length $wp^r$ and $(2w-1)p^r$.

\begin{theorem}\label{thm:mixed-wpr-optimal}
Let $p$ be a prime and $w<w^*$ be positive integers such that $p-1$ is divisible by $2w^*-2$.
If there is a code in $\CACe(p,w^*)$ with $(p-1)/(2w^*-2)$ codewords and the condition in \eqref{eq:quadratic-non-residue} holds, then for any integer $r\geq 1$, 
\begin{align*}
K\left(wp^r, w; w^*,\frac{p^r-1}{2w^*-2}\right)=\frac{p^r+1}{2}+\frac{p^r-1}{2w^*-2}.
\end{align*}
\end{theorem}
\begin{proof}
By setting $m=(p-1)/(2w^*-2)$ in Corollary~\ref{cor:mixed-wpr}, there exists a mixed-weight code in $\CAC(wp^r,\{w,w^*\})$ containing $(p^r+1)/2$ codewords of weight $w$ and $(p^r-1)/(2w^*-2)$ codewords of weight $w^*$.

Let $\C$ be any mixed-weight CAC of length $wp^r$ with weight-set $\{w,w^*\}$ having $(p^r-1)/(2w^*-2)$ codewords of weight $w^*$.
It suffices to show that $|\C|\leq (p^r+1)/2+ (p^r-1)/(2w^*-2)$.

Let $\mathcal{E}\subseteq\C$ be the collection of all exceptional codewords, and denote  by $H_S=\mathsf{H}(d(S))$ and $H_S^*=H_S\setminus\{0\}$ for $S\in\mathcal{E}$.
Consider any $S\in\mathcal{E}$.
Note that $|S|\leq w^*$, and $|H_S|\leq 2|S|-2<p$ due to Corollary~\ref{cor:excpetional-bound} and the assumption that $p-1$ is divisible by $2w^*-2$.
By the same argument as in the derivation of \eqref{eq:main_wpr_case2-1}--\eqref{eq:main_wpr_case2-3}, either $|S|=w$ or $w^*$, we have $|H_S|\big| w$, $|d^*(S)|\geq 2|S|-2-|H^*_S|$ and $\sum_{S\in\mathcal{E}}|H^*_S|\leq w-1$.
This concludes that 
\begin{align*}
\sum_{S\in\mathcal{E}}|d^*(S)| &\geq \sum_{S\in\mathcal{E}}\left(2|S|-2\right) - \sum_{S\in\mathcal{E}}|H^*_S| \\ 
&\geq |\mathcal{E}_{w^*}|(2w^*-2) + |\mathcal{E}_{w}|(2w-2) -(w-1),
\end{align*}
where $\mathcal{E}_{w^*}$ and $\mathcal{E}_w$ denote the sets of codewords in $\mathcal{E}$ with weights $w^*$ and $w$, respectively.
By the disjoint-difference-set property,
\begin{align*}
wp^r-1 = |\mathbb{Z}^*_{wp^r}| &\geq \sum_{S\in\C\setminus\mathcal{E},|S|=w^*}|d^*(S)| + \sum_{S\in\C\setminus\mathcal{E},|S|=w}|d^*(S)| + \sum_{S\in\mathcal{E}}|d^*(S)|\\
&\geq \left(\frac{p^r-1}{2w^*-2}-|\mathcal{E}_{w^*}|\right)(2w^*-2) + \left(|\C|-\frac{p^r-1}{2w^*-2}-|\mathcal{E}_w|\right)(2w-2) \\
& \quad + |\mathcal{E}_{w^*}|(2w^*-2) + |\mathcal{E}_{w}|(2w-2) -(w-1) \\
& \ \geq p^r-1 + \left(|\C|-\frac{p^r-1}{2w^*-2}\right)(2w-2) - (w-1),
\end{align*}
which implies that $|C|-(p^r-1)/(2w^*-2)\leq (p^r+1)/2$.
This completes the proof.
\end{proof}

\begin{theorem}\label{thm:mixed-2w-1pr-optimal}
Let $p$ be a prime and $w<w^*$ be positive integers such that $2w^*-2$ divides $p-1$ and $2w-1$ divides $w^*-1$ or $2w^*-1$.
If there is a code in $\CACe(p,w^*)$ with $(p-1)/(2w^*-2)$ codewords, then for any integer $r\geq 1$, 
\begin{align*}
K\left((2w-1)p^r, w; w^*,\frac{p^r-1}{2w^*-2}\right)=p^r+\frac{p^r-1}{2w^*-2}.
\end{align*}
\end{theorem}
\begin{proof}
By setting $m=(p-1)/(2w^*-2)$ in Corollary~\ref{cor:mixed-2w-1pr}, there exists a mixed-weight CAC in $\CAC((2w-1)p^r,\{w,w^*\})$ containing $p^r$ codewords of weight $w$ and $(p^r-1)/(2w^*-2)$ codewords of weight $w^*$.

Let $\C$ be any mixed-weight CAC of length $(2w-1)p^r$ with weight-set $\{w,w^*\}$ having $(p^r-1)/(2w^*-2)$ codewords of weight $w^*$.
It suffices to show $|\C|\leq p^r+ (p^r-1)/(2w^*-2)$.
Firstly, we claim that all codewords in $\C$ is non-exceptional.
Pick $S\in\C$.
Note that $|\HH(d(S))|$ divides $(2w-1)p^r$ since $\HH(d(S))$ is a subgroup of $\Z_{(2w-1)p^r}$.
When $|S|=w$, by Corollary~\ref{cor:excpetional-bound}, $|\HH(d(S))|\leq 2w-2<p$, which implies that $|\HH(d(S))|\big|(2w-1)$.
By Lemma~\ref{lem:nonexceptional}, $S$ is non-exceptional.
Similarly, we also have $|\HH(d(S))|\big|(2w-1)$ in the case when $|S|=w^*$.
By the assumption that $2w-1$ divides $w^*-1$ or $2w^*-1$, we further have $|\HH(d(S))|\big|(w^*-1)$ or $|\HH(d(S))|\big|(2w^*-1)$.
By Lemma~\ref{lem:nonexceptional} again, $S$ is non-exceptional.

Finally, by the disjoint-difference-set property,
\begin{align*}
(2w-1)p^r-1=|\Z^*_{(2w-1)p^r}| &\geq \sum_{S\in\C,|S|=w^*}|d^*(S)| + \sum_{S\in\C,|S|=w}|d^*(S)| \\
&\geq \left(\frac{p^r-1}{2w^*-2}\right)(2w^*-2) + \left(|\C|-\frac{p^r-1}{2w^*-2}\right)(2w-2).
\end{align*}
Hence, the result follows.
\end{proof}

\begin{remark}\rm
With Theorems~\ref{thm:mixed-w-1pr}--\ref{thm:mixed-2w-1pr-optimal}, we are ready to analyze how the proposed mixed-weight CACs provide heterogeneous throughput/delay performance.
Obviously, a user assigned a codeword of weight $w'\geq w$ is able to transmit at least $w'-w+1$ packets successfully during every $L$ consecutive slots, which leads to higher worst-case throughput for larger $w'$ if we employ erasure correcting coding across packets to recover data lost due to collisions\footnote{For example, each user uses a shortened Reed-Solomon (RS) code of length $w'$ over the finite field of size $Q\geq w'$ to code its $w'-w+1$ information packets into $w'$ transmitted packets in each frame of $L$ slots corresponding to its used codeword, so that the data lost can be recovered because of the maximal-distance separable (MDS) property of RS codes.}.  
On the other hand, we note that the distance of adjacent ones in every codeword of a mixed-weight CAC is not smaller than $L/p^r$, which takes different values for different cases.  
More specifically, we have $L/p^r=w-1$, $w$, and $2w-1$ for the cases stated in Theorem~\ref{thm:mixed-w-1pr} and Corollaries~\ref{cor:mixed-wpr}--\ref{cor:mixed-2w-1pr}, respectively.
So, a user assigned a codeword of weight $w'\geq w$ would enjoy the worst-case delay $L-(w'-w)L/p^r$ slots.
\end{remark}

\section{Conclusion}\label{sec:conclusion}

We generalize some previously known constructions of constant-weight CACs in various aspects and propose several classes of optimal CACs.
Firstly, a direct construction of CACs of length $\frac{w-1}{d}p^r$ with weight $w$ is proposed in Theorem~\ref{thm:construction_direct} by the help of some properties of cosets in Group Theory.
By some techniques in Additive Combinatorics and Kneser's Theorem, the obtained CACs are proved to be optimal in Theorem~\ref{thm:main_w-1dpr}.
As an application of Theorem~\ref{thm:main_w-1dpr}, we provide several series of optimal CACs in Corollaries~\ref{coro:w-1prw} -- \ref{coro:3pr7} by Gauss's Lemma and the Law of Quadratic Reciprocity.
Secondly, constructions of CACs of length $p^r,wp^r$ and $(2w-1)p^r$ by extending smaller-length CACs are given in Theorems~\ref{thm:construction_pr}, \ref{thm:construction_wpr} and \ref{thm:construction_2w-1pr}, respectively.
Sufficient conditions of the constructed CACs to be optimal are characterized in Theorems~\ref{thm:main_pr} -- \ref{thm:main_2w-1pr}.
Finally, we study mixed-weight CACs for the first time for the purpose of increasing the throughput and deducing the access delay of some potential users with higher priority.
As an application of the proposed direct construction of CACs given in Theorem~\ref{thm:construction_direct}, we in Theorem~\ref{thm:mixed-w-1pr} provide a general construction of mixed-weight CACs of length $(w-1)p^r$ consisting of three or more distinct weights.
With some specific parametric requirements, we obtain a series of optimal mixed-weight CACs containing two different weights in Theorem~\ref{thm:mixed-w-1pr-optimal}.
Two classes of optimal mixed-weight CACs of length $wp^r$ and $(2w-1)p^r$ are respectively given in Theorems~\ref{thm:mixed-wpr-optimal} and \ref{thm:mixed-2w-1pr-optimal} as well.
We also analyze the heterogeneous throughput/worst-case delay supported by mixed-weight CACs based on Theorems~\ref{thm:mixed-w-1pr}--\ref{thm:mixed-2w-1pr-optimal}.
It is worth noting that the heterogeneous average delay performance can be easily analyzed using the method in~\cite[Theorem 1]{WSLWS13}.

\section*{Acknowledgment}
The authors would like to thank the anonymous reviewers and the Editor for their valuable comments that improved the presentation of this article.


\end{document}